\documentclass[journal]{IEEEtran}
\usepackage{url}            
\usepackage{booktabs}       
\usepackage{amsfonts}       
\usepackage{nicefrac}       
\usepackage{microtype}      
\usepackage{lipsum}
\usepackage{amsmath}
\usepackage{graphicx}
\usepackage{subcaption}
\usepackage{graphicx}
\usepackage{amssymb}
\usepackage{pdfpages}
\usepackage{float}
\usepackage{subfiles} 
\usepackage{amsthm}
\usepackage{mathtools}
\usepackage{mathrsfs}
\newtheorem{theorem}{Theorem}

\newtheorem{definition}{Definition}
\newtheorem{remark}{Remark}

\newtheorem{assumption}{Assumption}
\newtheorem{problem}{Problem}
\DeclarePairedDelimiter{\diagfences}{(}{)}
\newcommand{\diag}{\operatorname{diag}\diagfences}
\DeclareUnicodeCharacter{221E}{-}

\DeclarePairedDelimiter{\Intfences}{(}{)}
\newcommand{\Int}{\operatorname{Int}\Intfences}
\usepackage{algorithm}
\usepackage{algpseudocode}
\newcommand\ignore[1]{{}}
\algrenewcommand\algorithmicrequire{\textbf{Input:}}
\algrenewcommand\algorithmicensure{\textbf{Output:}}
\newcommand{\dx}[1]{\textcolor{red}{#1}}

\usepackage{hyperref}
\usepackage{cleveref}

\begin{document}
\title{Off-Policy Risk-Sensitive Reinforcement Learning Based Constrained Robust Optimal Control}

\author{Cong Li,~\IEEEmembership{}
        Qingchen Liu{$^*$}~\IEEEmembership{}, Zhehua Zhou,~\IEEEmembership{} 
        Martin Buss,  and~Fangzhou Liu
\thanks{C. Li, Z. Zhou, M.Buss, and F. Liu are with the Chair of Automatic Control Engineering, Technical University of Munich, Theresienstr. 90,80333, Munich, Germany e-mail: \{cong.li,  zhehua.zhou, mb, fangzhou.liu\}@tum.de.}
\thanks{Q. Liu is with the Chair of Information-Oriented Control, Technical University of Munich, Barer Stra 21, 80333, Munich, Germany e-mail: liuqingchen1989@gmail.com.}
\newline
\dx{This work has been submitted to the IEEE for possible publication. Copyright may be transferred without notice, after which this version may no longer be accessible.}
}

\maketitle

\begin{abstract}
This paper proposes an off-policy risk-sensitive reinforcement learning based control framework for stabilization of a continuous-time nonlinear system that subjects to additive disturbances, input saturation, and state constraints. By introducing pseudo controls and risk-sensitive input and state penalty terms, the constrained robust stabilization problem of the original system is converted into an equivalent optimal control problem of an auxiliary system. Then, aiming at the transformed optimal control problem, we adopt adaptive dynamic programming (ADP) implemented as a single critic structure to get the approximate solution to the value function of the Hamilton-Jacobi-Bellman (HJB) equation, which results in the approximate optimal control policy that is able to satisfy both input and state constraints under disturbances. By replaying experience data to the off-policy weight update law of the critic artificial neural network, the weight convergence is guaranteed. Moreover, to get experience data to achieve a sufficient excitation required for the weight convergence, online and offline algorithms are developed to serve as principled ways to record informative experience data. 
The equivalence proof demonstrates that the optimal control strategy of the auxiliary system robustly stabilizes the original system without violating input and state constraints. The proofs of system stability and weight convergence are provided. Simulation results reveal the validity of the proposed control framework.
\end{abstract}


 \section{Introduction}
    \IEEEPARstart{R}{ecently} adaptive dynamic programming (ADP), which emerges as a successful implementation of reinforcement learning (RL) in the control field, has been proposed to approximately solve regulation or tracking problems of continuous or discrete time nonlinear systems with performance and/or robustness requirements, see \cite{liu2020editorial,lewis2009reinforcement}, and the references therein. The artificial neural network (NN) based approximate solution to the value function of the Hamilton-Jacobi-Bellman (HJB) equation facilitates the approximate optimal control strategy under an actor-critic structure. However, although traditional ADP has been widely adopted to tackle performance and robustness problems, input and state constraint satisfaction during the learning process, which is mainly investigated for safety concerns (e.g, restrictions on torques, joint angels and angular velocities of robot manipulators), has not yet been efficiently addressed. Violations of any of them could lead to possible serious consequences such as damage to physical components. This motivates us to develop a constraint-satisfying ADP-based control strategy. Furthermore, for practical applications, the constraint-satisfying control strategy should guarantee desired performance even under an uncertain environment.
    
    \subsection{Prior and related works}
    \emph{Constrained ADP:} 
    In previous ADP related works \cite{abu2005nearly,na2018nonlinear}, by incorporating a suitable nonquadratic functional into the cost function, control limits have been addressed under an actor-critic structure, which leads to a bounded approximate optimal control policy in the form of a bounded hyperbolic tangent function. However, comparing with traditional ADP where a quadratic functional is usually adopted to represent the desired performance regarding control efforts \cite{vamvoudakis2010online}, the introduced nonquadratic functional leads to an inevitable performance compromise problem that is ignored in existing related works \cite{abu2005nearly,na2018nonlinear}.
    In terms of restrictions on system states, most of previous ADP related works adopt the system transformation technique to deal with state constraints under an actor-critic structure \cite{yang2019online,he2019data,sun2018disturbance}. This method seeks for appropriate variable transformations that enable transformed system states to approach to infinity when potential state constraint violation happens.
    Therefore, approximate optimal control strategies that could achieve bounded transformed system states are constraint-satisfying control policies of the original system.
    The system transformation technique, nonetheless, is limited to simple constraint forms, e.g., restricted working space in a rectangular form. For certain state constraints, it may be infeasible to find suitable variable transformations to convert the constrained problem into an unconstrained counterpart.
    Besides, although general state constraints could be tackled by the well-designed penalty functions \cite{na2018nonlinear,abu2006nonlinear}, which become dominant in the optimization process when possible constraint violation happens and thus punish potential dangerous behaviours, 
    no strict constraint satisfaction proofs are provided. However, in certain cases such as human-robot interaction, even the violation of constraints regarding safety issues in a small possibility is not allowed. Moreover, the trade-off between constraint satisfaction and performance is ignored in \cite{na2018nonlinear,abu2006nonlinear}, wherein control strategies that partially focus on constraint satisfaction or performance often result in a lack of practicability \cite{slotine1991applied}. 

    \emph{Single critic structure:} For actor-critic structure adopted in the aforementioned ADP related works,
    the interplay between actor and critic NN is likely to cause instability because a wrong step taken by either of an actor or critic learning agent might adversely affect the other and destabilize the learning process \cite{parisi2020reinforcement}. 
    To solve this potential instability problem, a single critic structure is adopted in \cite{heydari2012finite} where the approximated value function from the critic NN is directly used to construct the approximate optimal control policy. 
    The prerequisite of straightly using the approximated value function to construct the control policy is that the weight convergence of the estimated critic NN to the actual value is guaranteed.
    However, to the best of our knowledge, no strict proofs and/or simulation results have been provided to declare the guaranteed weight convergence to the actual value in most of single critic structure related works \cite{heydari2012finite,yang2018event,padhi2006single}. 
    Conventionally, the weight convergence is checked by the persistence of excitation (PE) condition \cite{boyd1986necessary}.
    Among existing ADP related works \cite{vamvoudakis2010online,bhasin2013novel}, the PE condition has been satisfied by directly adding external noises to inputs to achieve a sufficient exploration of the operation space. However, this method is unfavorable to practical applications given that the direct incorporation of external noises into control inputs may degrade control performance and cause nuisance, waste of energy, etc.
    Furthermore, in this method, the choice of explicit noise forms and the time to remove them during the online learning process highly depend on prior knowledge. 
    Hence, for the guaranteed weight convergence, a more practical and easily implemented method to satisfy the PE condition is needed.

    \emph{Experience replay:} Given the fact that the satisfaction of the PE condition implies that available data regarding unknown weights to be learned are rich enough during the entire learning period \cite{ioannou2012robust}, a feasible way to get the desired rich data is to reuse past data generated during the learning process. Recently experience replay (ER) emerges as an effective data-generating mechanism that replays experience data to accelerate the online learning process \cite{lin1992self}. 
    In ADP related works \cite{yang2019online,kamalapurkar2016model}, a fixed number of recently recorded transitions are replayed to actor or critic learning agents to avoid the necessity of using external noises to satisfy the required PE condition for the weight convergence.
    However, for this sequent way of data usage, experience data with different richness levels are used without discrimination, and partial informative data (e.g., data from initial exploring phases) are only used once then abandoned immediately. These characteristics of data usage result in poor sample efficiency and the collected data might not be rich enough for the weight convergence.
    Such a sample deficiency problem also exists in uniform sampling based ER techniques where all transitions are replayed at a same frequency regardless of their significance \cite{fedus2020revisiting}.
    Unlike replaying experience data in a sequent or an uniform sampling way, prioritized experience replay (PER) is a more efficient technique that prioritizes data according to certain criteria \cite{schaul2015prioritized}.
    Although all the aforementioned ER techniques have shown promise to accelerate learning, the implementation often accompanies with prior experience and extensive parameter tuning \cite{yang2019online,kamalapurkar2016model}. To the best of our knowledge, there exists no related works on principled ways to provide rich enough experience data to accelerate the online leaning process.
    \subsection{Contribution}
    This work focuses on an off-policy risk-sensitive RL-based control framework, as summarized in Fig.\ref{algorithm framework}, for the control task of a nonlinear system under additive disturbances, input and state constraints. 
    Based on pseudo controls, risk-sensitive input and state penalty terms, we transform a generally intractable constrained robust stabilization problem into an optimal control problem, which is then approximately solved by ADP under a single critic structure with an off-policy weight update law.
    The contribution of this paper is three-fold. 
    First, we propose novel risk-sensitive state penalty terms to act as risk criteria during the learning process, which enables us to tackle state constraints in a long time-horizon, and preserve performance with strict constraint satisfaction proofs.
    Second, by exploiting experience data, we design an efficient off-policy critic NN weight update law that guarantees weight convergence without causing undesirable oscillations and additional control effort expenditures. 
    Third, principled ways implemented as online and offline experience buffer construction algorithms are proposed to provide the required rich enough experience data for the weight convergence.
   
    The remainder of this article is organized as follows. Section \ref{Section 2} provides the formulation of a constrained robust stabilization problem, the transformation to an optimal control problem, and the problem equivalence proof. 
    Section \ref{Section 3} elucidates the approximated solution to the optimal control problem and the off-policy weight update law. Besides, the critic NN weight convergence proof and the system stability proof are provided.
    Simulation results shown in Section \ref{Section 4} illustrate the effectiveness of the proposed control framework.
    Finally, Section \ref{Section  5} concludes this paper.
    
    \emph{Notations:} Throughout this paper, $\mathbb{R}$ ($\mathbb{R}^{+}$) denotes the set of real (positive) numbers; $\mathbb{N}^{+}$ denotes the set of positive natural numbers;
    $\mathbb{R}^{n}$ is the Euclidean space of $n$-dimensional real vector; $\mathbb{R}^{n \times m}$ is the Euclidean space of $n \times m$ real matrices; $I_{m \times m}$ represents the identity matrix with dimension $m \times m$; 
    $\Int{\mathbb{D}}$ and $\partial \mathbb{D}$ denote the interior and boundary of the set $\mathbb{D}$, respectively;
    $\lambda_{\min}(M)$ and $\lambda_{\max}(M)$ are the maximum and minimum eigenvalues of a symmetric matrix $M$, respectively; $\diag {a_{1},...,a_{n}}$ is the $n \times n$ diagonal matrix with the value of main diagonal as $a_{1},...,a_{n}$;
    The $i$-th entry of a vector $x = [x_{1},...,x_{n}]^{\top}\in \mathbb{R}^{n}$ is denoted by $x_{i}$, and $\left\| x \right\| = \sqrt{\sum_{i=1}^{N}|x_{i}|^2}$ is the Euclidean norm of the vector $x$;
    The $ij$-th entry of a matrix $D \in \mathbb{R}^{n \times m}$ is denoted by $d_{ij}$, and $\left\|D\right\| = \sqrt{\sum_{i=1}^{n}\sum_{j=1}^{m}|d_{ij}|^2}$ is the Frobenius norm of the matrix $D$. For notational brevity, time-dependence is suppressed without causing ambiguity. 
    \begin{figure}[!t]  
    \centerline{\includegraphics[width=3.4in]{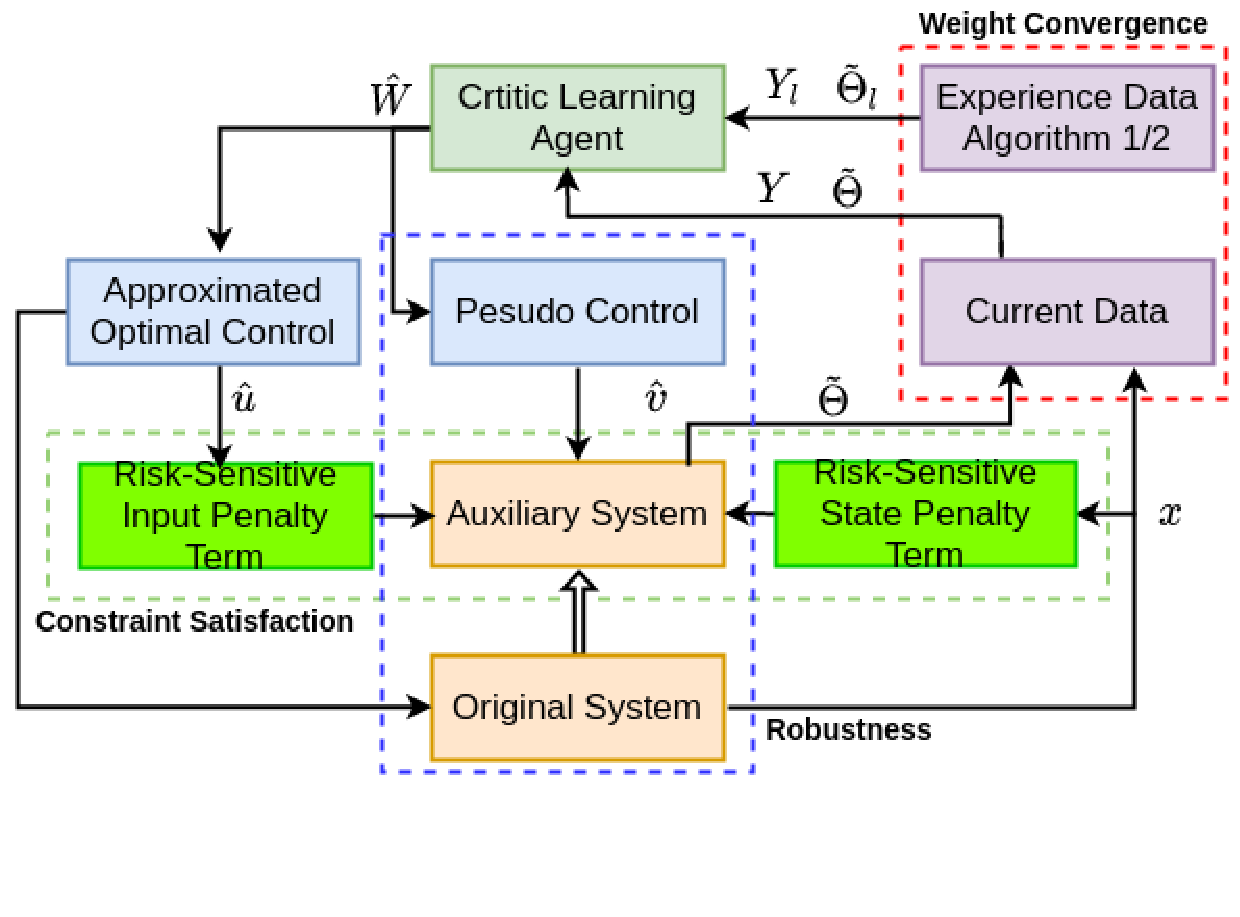}}
    \caption{Schematic of the proposed off-policy risk-sensitive RL-based control framework that contains three key components: 1) Robustness: pseudo controls to address additive disturbances in Section \ref{section PT}; 2) Constraint satisfaction: risk-sensitive input/state penalty terms to enforce input/state constraint satisfaction in Section \ref{safety and performance trade off}; 3) Weight convergence: rich enough experience data to accelerate online learning in Section \ref{Section 3}.}
    \label{algorithm framework}
    \end{figure}
    
    \section{Problem formulation} \label{Section 2}
    \subsection{Formulation of constrained robust stabilization problem} \label{section PF}
    Consider the continuous-time nonlinear dynamical system:
    \begin{equation} \label{original sys}
        \dot{x} = f(x)+g(x)u(x)+k(x)d(x),
    \end{equation}
    where $x \in \mathbb{R}^{n}$ and $u(x) \in \mathbb{R}^{m}$ are states and inputs of the system. $f(x) : \mathbb{R}^{n} \to \mathbb{R}^{n}$, $g(x) : \mathbb{R}^{n} \to \mathbb{R}^{n \times m}$ are the known drift and input dynamics, respectively. $k(x): \mathbb{R}^{n} \to \mathbb{R}^{n \times r}$ represents the known differential system function. $d(x): \mathbb{R}^{n} \to \mathbb{R}^{r}$ denotes the unknown additive disturbance. The general case that the additive disturbance is unmatched, i.e., $k(x) \neq g(x)$, is considered. Assuming that $f(0) = 0$ and $d(0)=0$, which means that the equilibrium point is $x = 0$.

    Before proceeding, the following assumptions are provided, which are common in ADP related works. 
        \begin{assumption} \cite{abu2005nearly} \label{bound of fg}
        $f(x)+g(x)u$ is Lipschitz continuous on a set $\Omega \subseteq \mathbb{R}^{n}$ that contains the origin, and the system is stabilizable on $\Omega$. 
        There exists $g_{M} \in \mathbb{R}^{+}$ such that the input dynamics is bounded by $\left\| g(x) \right\| \leq g_{M}$.
    \end{assumption}
    \begin{assumption} \cite{lin1998optimal} \label{bound of d and inverse term}
    The unknown additive disturbance $d(x)$ is bounded by a known nonnegative function $d_{M}(x)$: $\left\|d(x)\right\| \leq d_{M}(x)$, and $d_{M}(0) =0$. 
    \end{assumption}

    Based on the aforementioned settings, we formulate the constrained robust stabilization problem (CRSP) as follows.
    \begin{problem} [CRSP] \label{Robust constraint control problem}
    Given  Assumptions \ref{bound of fg}-\ref{bound of d and inverse term},
    design a control strategy $u(x)$
    such that the closed-loop system \eqref{original sys} that subjects to additive disturbances
    is stable under input saturation
    \begin{equation} \label{input saturation}
    \mathbb{U}_j = \left\{u_j \in \mathbb{R} : \left| u_{j} \right| \leq \beta \right\}, j = 1, \cdots, m, 
    \end{equation}
    where $\beta \in \mathbb{R}^{+}$ is a known saturation bound;
    and state constraints
    \begin{equation} \label{state constriant}
        \mathbb{X}_i = \left\{x \in \mathbb{R}^{n} : h_{i}(x) < 0 \right\}, i = 1, \cdots, n_c,
    \end{equation}
    where $\mathbb{X}_i$ is a closed and convex set that contains the origin in its interior; $h_{i}(x):\mathbb{R}^{n} \to \mathbb{R}$ is a known continuous function that relates with the $i$-th state constraint; $n_c \in \mathbb{N}^{+}$ is the number of considered state constraints.
    \end{problem}
    \subsection{Transformation to optimal control problem} \label{section PT}
    Problem \ref{Robust constraint control problem} can be decomposed into three sub-problems: disturbance rejection, input saturation and state constraint. It is not trivial for ADP to directly deal with these sub-problems together \cite{liu2020editorial}. 
    Thus, in this section, with pseudo controls proposed in \cite{lin1998optimal}, reformulated risk-sensitive input penalty terms based on \cite{abu2005nearly}, and our newly designed risk-sensitive state penalty terms, we firstly transform the CRSP clarified as Problem \ref{Robust constraint control problem} into an equivalent optimal control problem, and then attempt to solve the aforementioned sub-problems simultaneously under an optimization framework.
    
    \subsubsection{Pseudo control and auxiliary system}
    To address the additive disturbance under an optimization framework, firstly, by following \cite{lin1998optimal}, we decompose $k(x)d(x)$ as a sum of matched and unmatched disturbance elements
    \begin{equation} \label{uncertainty decomposition}
        k(x)d(x) = g(x)g^{\dagger}(x)k(x)d(x)+h(x)d(x),
    \end{equation}
    where $h(x) = (I-g(x)g^{\dagger}(x))k(x) : \mathbb{R}^{n} \to \mathbb{R}^{n \times r}$, and $\dagger$ denotes the Moore-Penrose inverse.
    \begin{assumption}\cite{lin1998optimal}\label{bound of decomposed disturbance part}
    The continuous function $h(x)$ is bounded as $\left\|h(x)\right\| \leq h_{M}$;
      $g^{\dagger}(x)k(x)d(x)$ is bounded by a nonnegative function $l_{M}(x)$ : $\left\|g^{\dagger}(x)k(x)d(x)\right\| \leq l_{M}(x)$, and $l_{M}(0) =0$.    
    \end{assumption}
    Then, a pseudo control $v(x): \mathbb{R}^{n} \to \mathbb{R}^{r}$ is introduced here, which is used to tackle the unmatched disturbance $h(x)d(x)$ in \eqref{uncertainty decomposition}, and the resulting auxiliary system follows \cite{lin1998optimal}
    \begin{equation} \label{auxiliary system}
    \dot{x} = f(x)+g(x)u(x)+h(x)v(x).
    \end{equation}
    By focusing on the auxiliary system \eqref{auxiliary system} and incorporating the squared matched/unmatched disturbance bound (i.e, $l^2_M$ and $d^2_M$) into the cost function \eqref{cost fuction},  we can address the additive disturbance under an optimization framework. The corresponding proofs are provided later in Theorem \ref{theorem for equivalence}.
    \subsubsection{Risk-sensitive input and state penalty terms}
    To tackle input/state constraints under an optimization framework, here we follow the idea of risk-sensitive RL where multiple risk measures, e.g, high moment or conditional value at risk, are used to deal with constraints of Markov decision processes \cite{shen2013risk}.
    However, the available risk measures in the risk-sensitive RL field cannot guarantee strict constraint satisfaction and/or not efficient (even inappropriate) to address constraints of continuous nonlinear systems. 
    Thus, we propose risk-sensitive input penalty term (RS-IP) in Definition \ref{risk-sensitive CIPF} and risk-sensitive state penalty term (RS-SP) in Definition \ref{risk-sensitive SPF} as new risk measures during the learning process to enforce strict satisfaction of input/state constraints of continuous nonlinear systems. 
    \begin{definition}[RS-IP]\label{risk-sensitive CIPF}
     A continuous and differential function $\phi (u)$ is a risk-sensitive input penalty term if it has the following properties:
     
     (1) A bounded monotonic odd function with $\phi(0) =0$;
     
     (2) The first-order partial derivatives of $\phi (u)$ is bounded.
    \end{definition}
    Here the RS-IP term is a reformulation of the nonquadratic functional used in \cite{abu2005nearly,heydari2012finite} to confront input constraints.
    \begin{definition}[RS-SP]\label{risk-sensitive SPF}
     Given the closed region $\mathbb{X}_i$, $i = 1, \cdots, n_c$, defined as \eqref{state constriant}, a continuous scalar function $S_i(x) : \mathbb{X}_i \to \mathbb{R}$, $i = 1, \cdots, n_c$, is a risk-sensitive state penalty term if the following proprieties hold:
     
    (1) $S_i(0) = 0$, and $S_i(x) > 0, \forall x \ne 0$; 
     
     (2) $S_i(x) \to \infty$ if $x$ approaches $\partial \mathbb{X}_i$; 
     
     (3) For initial value $x(0) \in \Int{\mathbb{X}_i}$, there exists $s \in \mathbb{R}^{+}$ such that $S_i(x(t))\leq s, \forall t \geq 0$ along solutions of the dynamics.
    \end{definition}
    Comparing with similar works \cite{na2018nonlinear,abu2006nonlinear} that use state penalty functions  to tackle state constraints but without strict constraint satisfaction proofs,
    the proposed RS-SP term enables us to provide strict constraint satisfaction proofs in Theorem \ref{theorem for equivalence}.
    Here the novel RS-SP term is inspired by the so-called barrier Lyapunov function \cite{tee2009barrier}.
    The first point of Definition \ref{risk-sensitive SPF} denotes that $S_i(x)$ is an effective Lyapunov function candidate, which enables $S_i(x)$ to serve as part of Lyapunov function for the system stability proof. The last two points imply that $\inf_{x \to \partial \mathbb{X}_i }S_i(x) = \infty$ and $\inf_{x \in \Int{\mathbb{X}_i}}S_i(x) \geq 0$, which means that $S_i(x)$ serves as a barrier certificate for an allowable operating region $\mathbb{X}_i$. 
    \subsubsection{Optimal control problem}
    Based on the auxiliary system \eqref{auxiliary system} and Definitions \ref{risk-sensitive CIPF}-\ref{risk-sensitive SPF}, an equivalent optimal control problem (OCP) of CRSP in Problem \ref{Robust constraint control problem} is clarified as Problem \ref{Optimal control problem}. Comparing with traditional ADP that accomplishes partial objectives of performance, robustness, and input/state constraint satisfaction \cite{liu2020editorial,lewis2009reinforcement,abu2005nearly,na2018nonlinear}, the applied problem transformation here enables us to consider such multiple objectives together.
    \begin{problem} [OCP] \label{Optimal control problem}
    Given Assumptions \ref{bound of fg}-\ref{bound of decomposed disturbance part}, consider the auxiliary system \eqref{auxiliary system},
    find $u(x)$ and $v(x)$ to minimize the cost function
    \begin{equation}\label{cost fuction}
    V(x(t)) = \int_{t}^{\infty} r(x(\tau),u(x(\tau)),v(x(\tau)))\,d\tau,
    \end{equation}
    where the utility function $r(x,u(x),v(x)) = r_{d}(x) + \rho v^{\top}(x)v(x) + r_{c}(x,u(x))$ with $\rho \in \mathbb{R}^{+}$, $r_{d}(x) = l^{2}_{M}(x) + \rho d^{2}_{M}(x)$, and $r_{c}(x,u(x)) =  \mathcal{W}(u(x))+\mathcal{L}(x)$.
    The input penalty function $\mathcal{W}(u(x))$ is defined as 
    \begin{equation}\label{control penalty function}
    \mathcal{W}(u(x)) =  \sum_{j=1}^{m}2\int_{0}^{u_{j}}  \beta R_j \phi^{-1}(\vartheta_j / \beta)  \,d\vartheta_j,
    \end{equation}
    where $\phi (\cdot)$ is the RS-IP term in Definition \ref{risk-sensitive CIPF}; $R_j$ is the $j$-th diagonal element of a positive definite diagonal matrix $R \in \mathbb{R}^{m \times m}$.
    The state penalty function $\mathcal{L}(x)$ is defined as 
    \begin{equation}\label{state penalty function}
    \mathcal{L}(x) = x^{\top} Q x + \sum_{i=1}^{n_c} k_{i} S_i(x),
    \end{equation}
    where $Q \in \mathbb{R}^{n \times n}$ is a positive definite matrix; $k_{i}$ is the risk sensitivity parameter that follows $k_{i} = 1/(1+d_i^2)$, where $d_i$ is the distance from the state $x$ to the boundary of $h_{i}(x)$; $S_i(\cdot)$ is the RS-SP term in Definition \ref{risk-sensitive SPF} for $i$-th state constraint.
    \end{problem}
    
    Unlike ADP related works \cite{abu2005nearly,na2018nonlinear} that incorporate nonquadratic functionals to tackle input saturation but without considering performance regarding control efforts,
    by choosing suitable matrix $R$, $\mathcal{W}(u(x))$ in \eqref{control penalty function} could take into consideration of requirements for both control limits and control energy expenditures.
    Much more details are introduced in Section \ref{Mechanism of input penalty function}.
    The common used risk-neutral quadratic function $x^{\top}Qx$ \cite{liu2020editorial,lewis2009reinforcement} (capturing the desired state performance) is augmented with the newly designed weighted RS-SP term $\sum_{i=1}^{n_c} k_{i} S_i(x)$ (addressing multiple state constraints) to construct $\mathcal{L}(x)$ in \eqref{state penalty function}, which enables us to consider state-related performance and constraint satisfaction together. 
     The incorporation of $S_i(x)$ into $\mathcal{L}(x)$  deteriorates the desired performance represented by $x^{\top}Qx$.  Therefore, we propose the risk sensitivity parameter $k_i$, which relates with the distance from the constraint boundary, to specify the inevitable trade-off between the state-related performance and constraint satisfaction during the learning process. Note that this kind of  trade-off is ignored in existing related works \cite{na2018nonlinear,abu2006nonlinear}.
    The detailed mechanism of $\mathcal{L}(x)$ is illustrated in Section \ref{Mechanism of state penalty function}.

    \subsection{Mechanism of input and state penalty functions} \label{safety and performance trade off}
    The mechanism of $\mathcal{W}(u(x))$ and $\mathcal{L}(x)$ to enable the learning process to preserve performance without violating strict input/state constraint satisfaction is clarified here.
    \subsubsection{Mechanism of input penalty function $\mathcal{W}(u(x))$} \label{Mechanism of input penalty function}
    By Definition \ref{risk-sensitive CIPF}, the explicit form of the RS-IP term is chosen as $\phi (\cdot) = \tanh(\cdot)$ \cite{abu2005nearly}.
    Given the inevitable trade-off between constraint satisfaction and performance regarding control inputs, $\mathcal{W}(u(x))$ are designed to counter input constraints \eqref{input saturation} and approximate $u^{\top}\Bar{R}u$ (a common desired performance criterion for control efforts) simultaneously, where $\Bar{R} \in \mathbb{R}^{m \times m}$ is a positive definite matrix designed according to users' preferences. 
    \begin{figure*}[!t]
    \centering
    \subfloat[Plotting of $z_{1}=\beta R_j \tanh^{-1}(u_j/\beta)$, $R_j$ =1]{\includegraphics[width=2.4in]{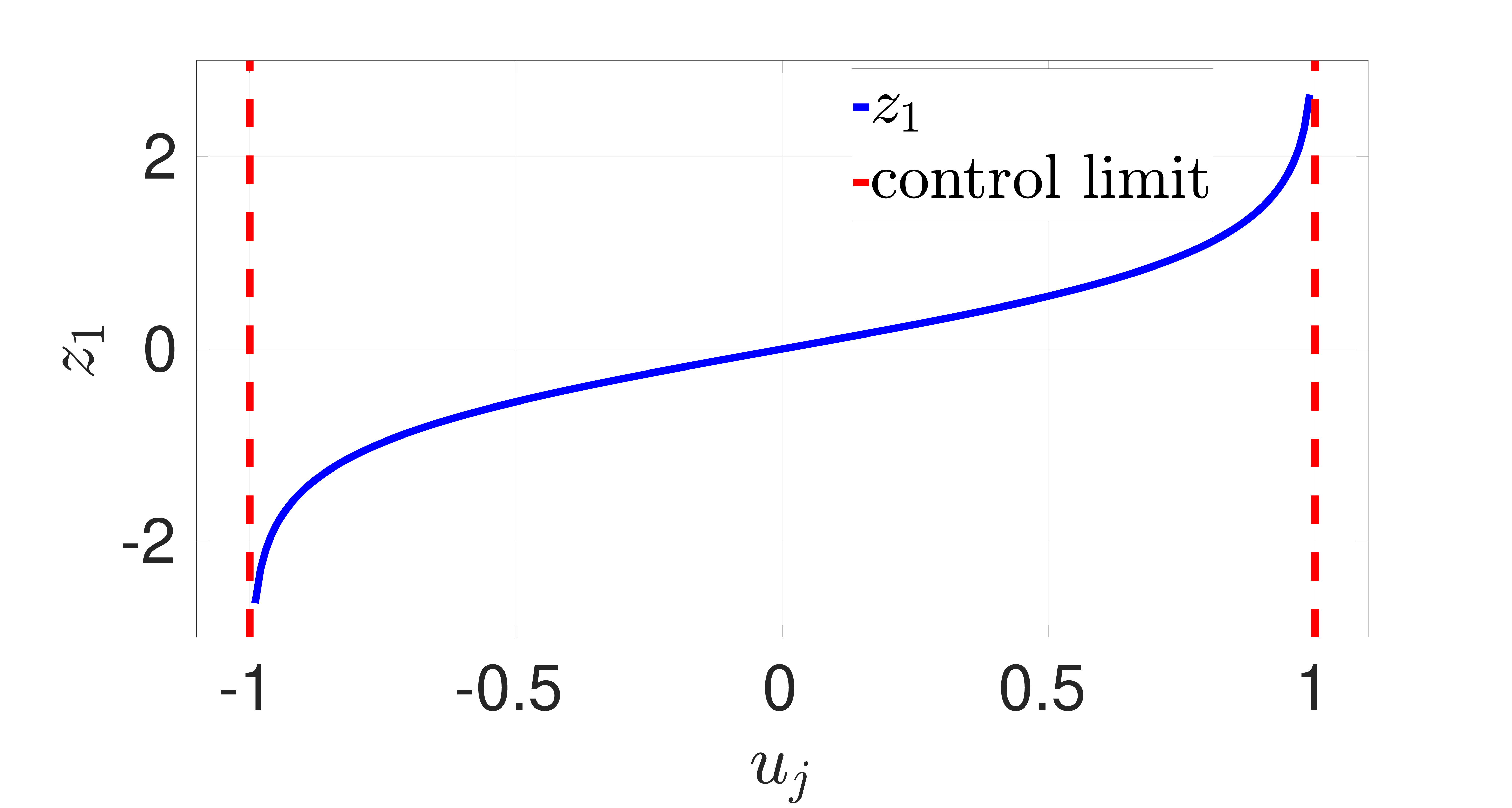}
    \label{Wuinverse}}
    \subfloat[Plotting of $u_j=-\beta \tanh(z_{2})$.]{\includegraphics[width=2.4in]{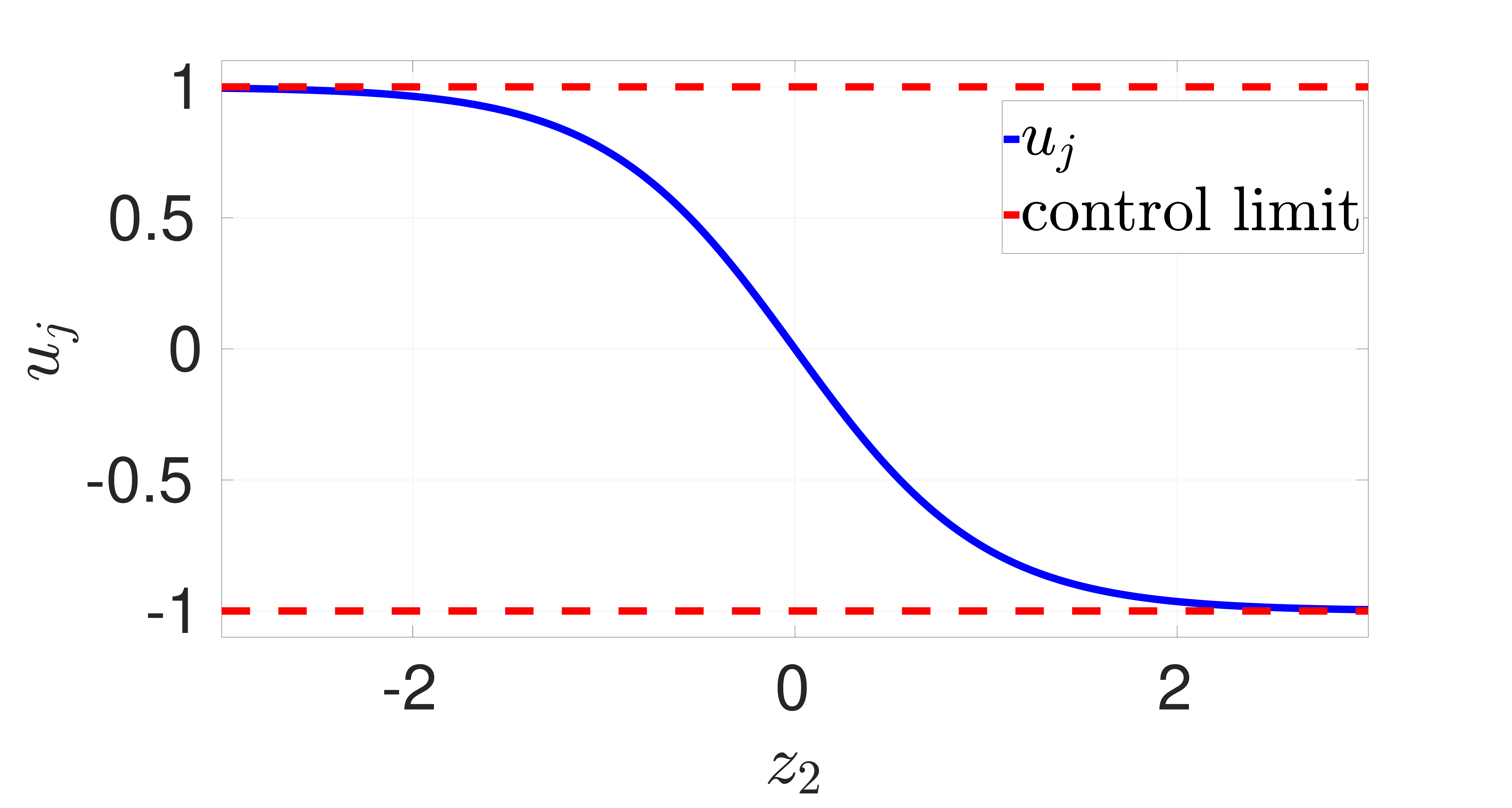}%
    \label{Wutanh}}
    \subfloat[Plotting of $u_j^{\top}\Bar{R}_j u_j$ and $\mathcal{W}_j(u_j)$ .]{\includegraphics[width=2.4in]{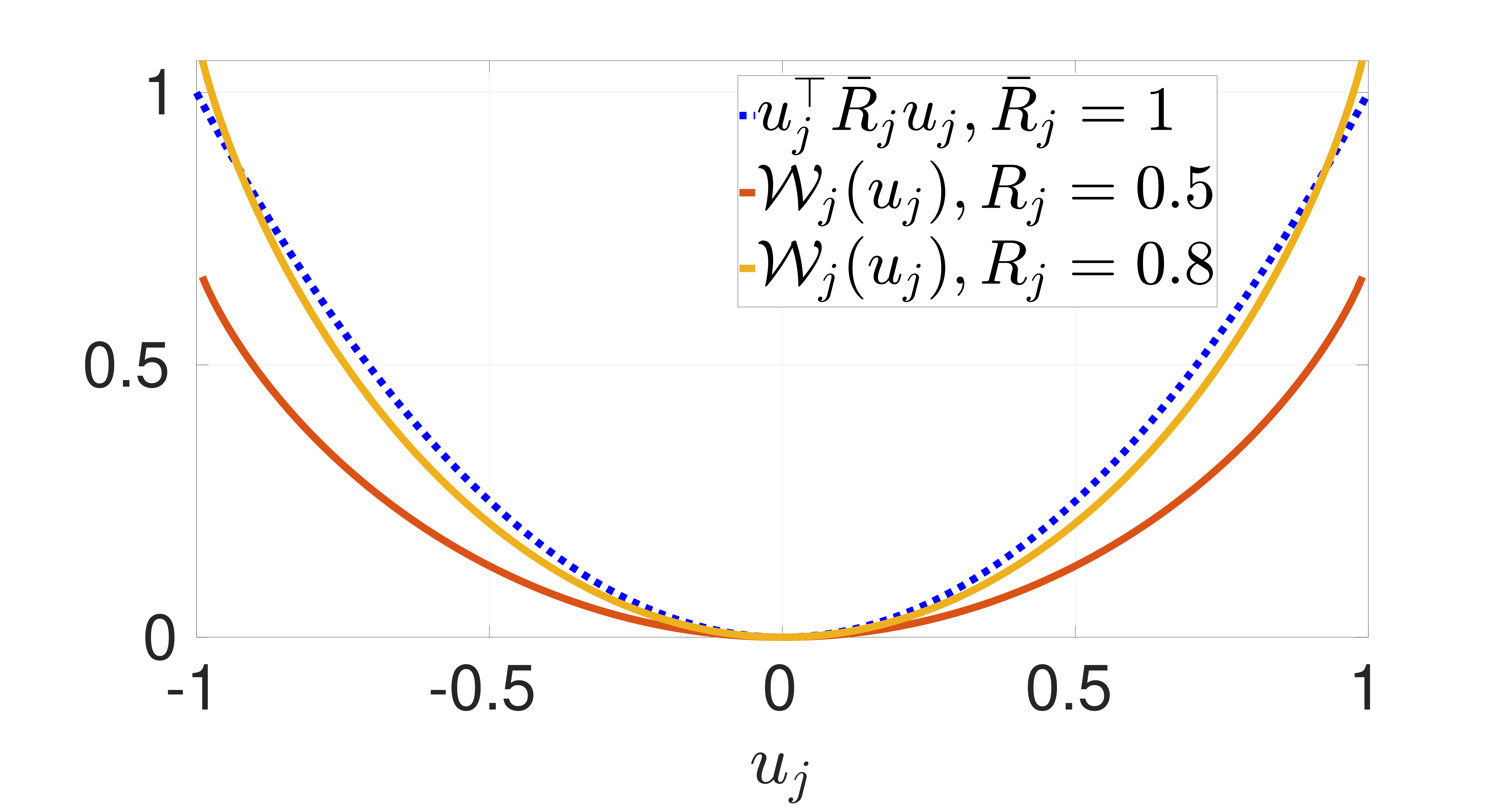}%
    \label{w and quadtratc u}}
    \caption{Graphical illustration of the working scheme of the input penalty function $\mathcal{W}(u(x))$ with $\beta = 1$, $z_{1}, z_{2} \in \mathbb{R}$, $u_j \in (-1,1)$.} 
    \label{w_u}
    \end{figure*}
    The mechanism of $\mathcal{W}(u(x))$ to tackle input constraints can be clarified from two perspectives, see Fig.\ref{Wuinverse} and Fig.\ref{Wutanh}, respectively. 
    In the first perspective, input constraints are considered in a long time-horizon. 
    $\mathcal{W}(u(x))$ in \eqref{control penalty function} is an integration of $\beta R_j\tanh^{-1}(u_j/\beta)$ that is denoted as $z_1$ in Fig.\ref{Wuinverse}.
    When any $u_j$, $j=1,\cdots,m$, approaches to the input constraint boundaries $\pm \beta$, it follows that the value of $\mathcal{W}(u(x))$ will be infinity. 
    Since the optimization process aims to minimize the cost function, the resulting optimal control strategy will be away from $\pm \beta$; Otherwise, a high value of the cost function occurs. 
    From the other perspective, according to the later result in \eqref{optimal u}, the resulting optimal control strategy based on $\mathcal{W}(u(x))$ is in a form of $\tanh(\cdot)$ whose boundness enforces strict satisfaction of input constraints, as shown in Fig.\ref{Wutanh}.
    The construction of $\mathcal{W}(u(x))$ to reflect the desired performance for control energy is shown in
    Fig.\ref{w and quadtratc u}. Consider $\mathcal{W}_j(u_j)$, the $j$-th summand of $\mathcal{W}(u(x))$. It follows that
    \begin{equation} \label{integration of W}
        \mathcal{W}_j(u_j) = 2\beta R_j u_j \tanh^{-1}(u_j/\beta)+\beta^2 R_j \log(1-u_j^2/\beta^2).
    \end{equation}
    As displayed in Fig.\ref{w and quadtratc u}, by adjusting the value of $R_j$, $\mathcal{W}_j(u_j)$ can approximate the desired control energy criterion $u^{\top}_j\Bar{R}u_j$ well. Based on the above discussion, we know that $\mathcal{W}(u(x))$ in \eqref{control penalty function} can tackle input constraints while preserving performance concerning control energy expenditures.
    \subsubsection{Mechanism of state penalty function $\mathcal{L}(x)$} \label{Mechanism of state penalty function}
    According to Definition \ref{risk-sensitive SPF}, when a potential state constraint violation happens, the corresponding RS-SP term will approach to infinity. Since the optimal control strategy aims to minimize the total cost, states will be pushed away from the direction where a high value of the RS-SP based $\mathcal{L}(x)$ occurs. Thus, the state constraint violation is avoided. 
    To satisfy Definition \ref{risk-sensitive SPF}, we choose $S_i(x) = \log{(h_{i}(x))}$ here.
    Note that the explicit form of $\log(h_{i}(x))$ is adjusted based on given state constraints, which is exemplified later. 
    For a better explanation of the mechanism of the RS-SP term $S_i(x)$ and the corresponding risk sensitivity parameter $k_{i}$, we present a four-dimensional system example with safe region defined as
    $\mathbb{X}_1 = \left\{x_1,x_2 \in \mathbb{R} : h_1(x_1,x_2) = x^2_{1}+x^2_{2}-1 < 0 \right\}$ \cite{wang2018safe}, $\mathbb{X}_2 = \left\{x_3 \in \mathbb{R} : h_2(x_3) = \left| x_3\right|-2 < 0 \right\}$, and $\mathbb{X}_3= \left\{x_4 \in \mathbb{R} : h_3(x_4) = \left| x_4 \right|-3 < 0 \right\}$ \cite{yang2019online}.
    The corresponding RS-SP terms are designed as $S_1(x_1,x_2) = \log (\alpha(x_{2})/(\alpha(x_{2})-x^2_{1}))$ with $\alpha(x_{2}) = 1-x^2_{2}$, $S_2(x_3) = \log (4/(4-x^2_{3}))$, and $S_3(x_4)=\log (9/(9-x^2_{4}))$, respectively.
     As displayed in Fig.\ref{circle RS-SP}-\ref{rectangle RS-SP}, these RS-SP terms act as barriers at constraint boundaries and confine states remain in the safe regions. This inherent risk-sensitive property enables us to tackle state constraints under an optimization framework.
    As long as initial states lie in the safe regions and the cost function is always bounded as time evolves, the subsequent state evolution will be restricted to the safe regions.
     From Fig.\ref{circle RS-SP k}-\ref{rectangle RS-SP k}, we know that the role of $S_i(x)$ will be discouraged by $k_{i}$ when states are far away from the boundary of $h_{i}(x)$. Therefore, performance regrading states is maintained when no state constraint violation occurs. 
    \begin{figure*}[!t]
    \centering
    \subfloat[Plotting of $S_1(x_1,x_2)$.]{\includegraphics[width=3.5in]{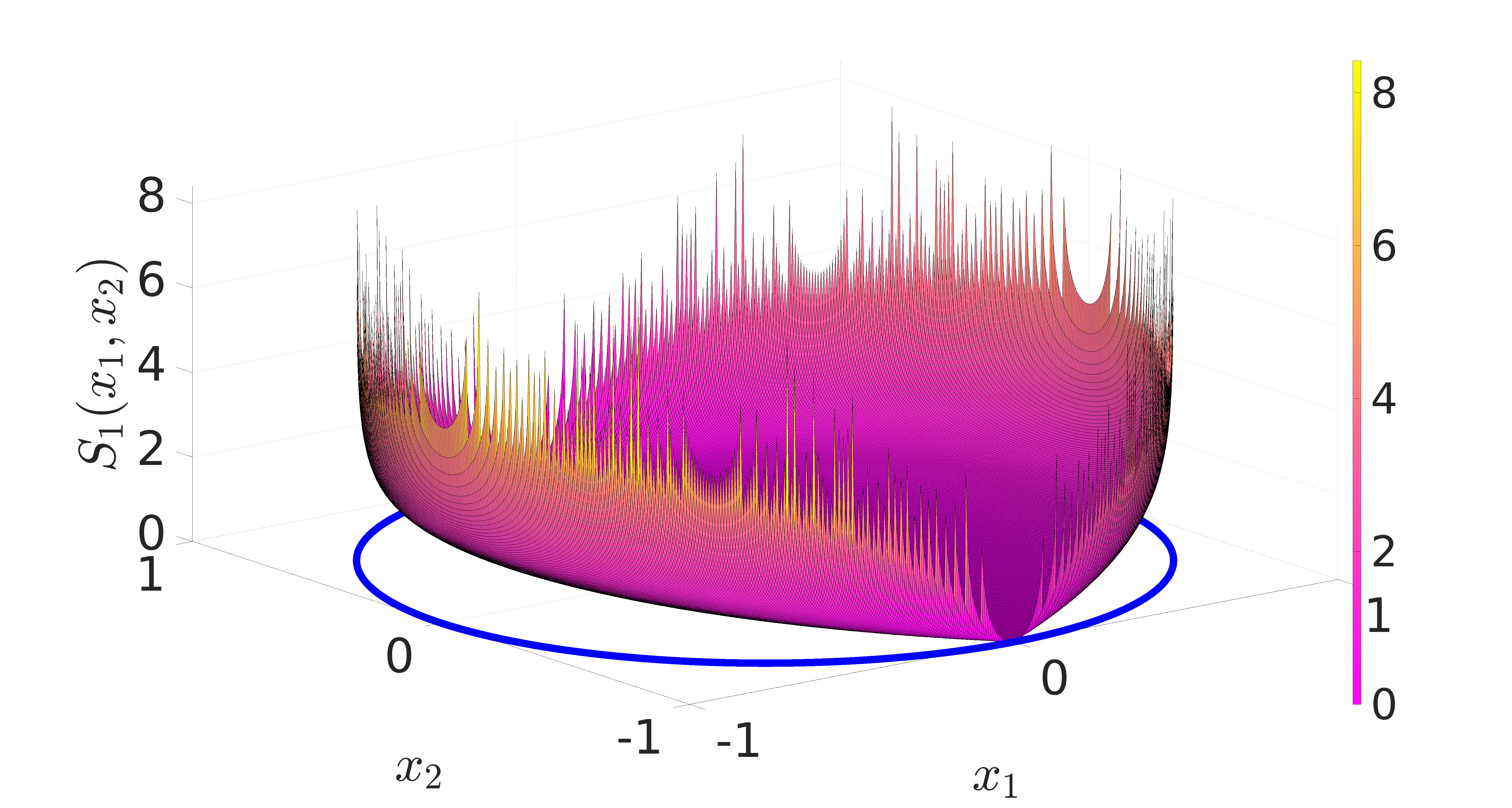}%
    \label{circle RS-SP}}
    \subfloat[Plotting of the sum of $S_2(x_3)$ and $S_3(x_4)$.]{\includegraphics[width=3.5in]{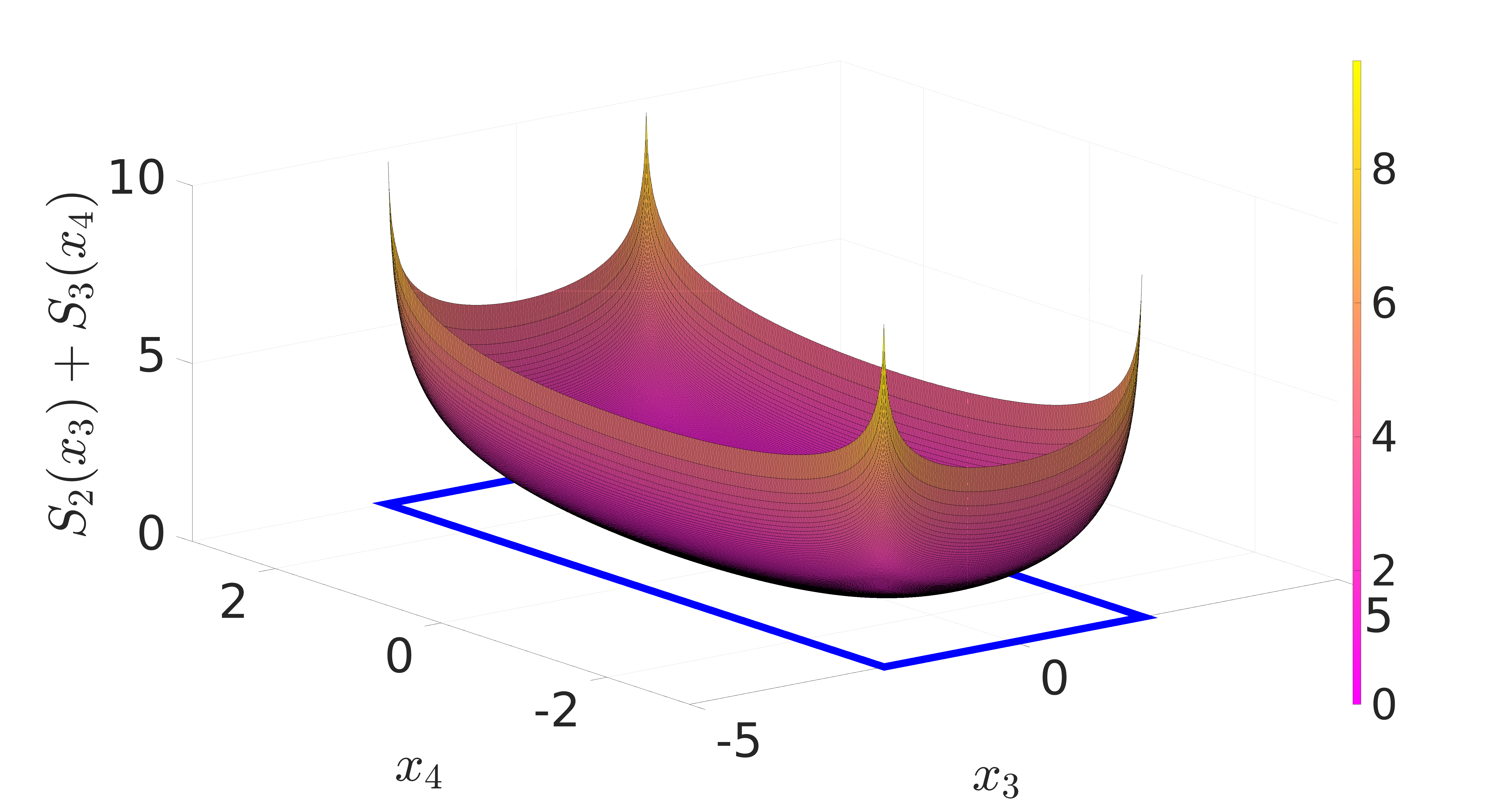}%
    \label{rectangle RS-SP}}
    \\
    \subfloat[Plotting of $k_1$ for $S_1(x_1,x_2)$.]{\includegraphics[width=3.5in]{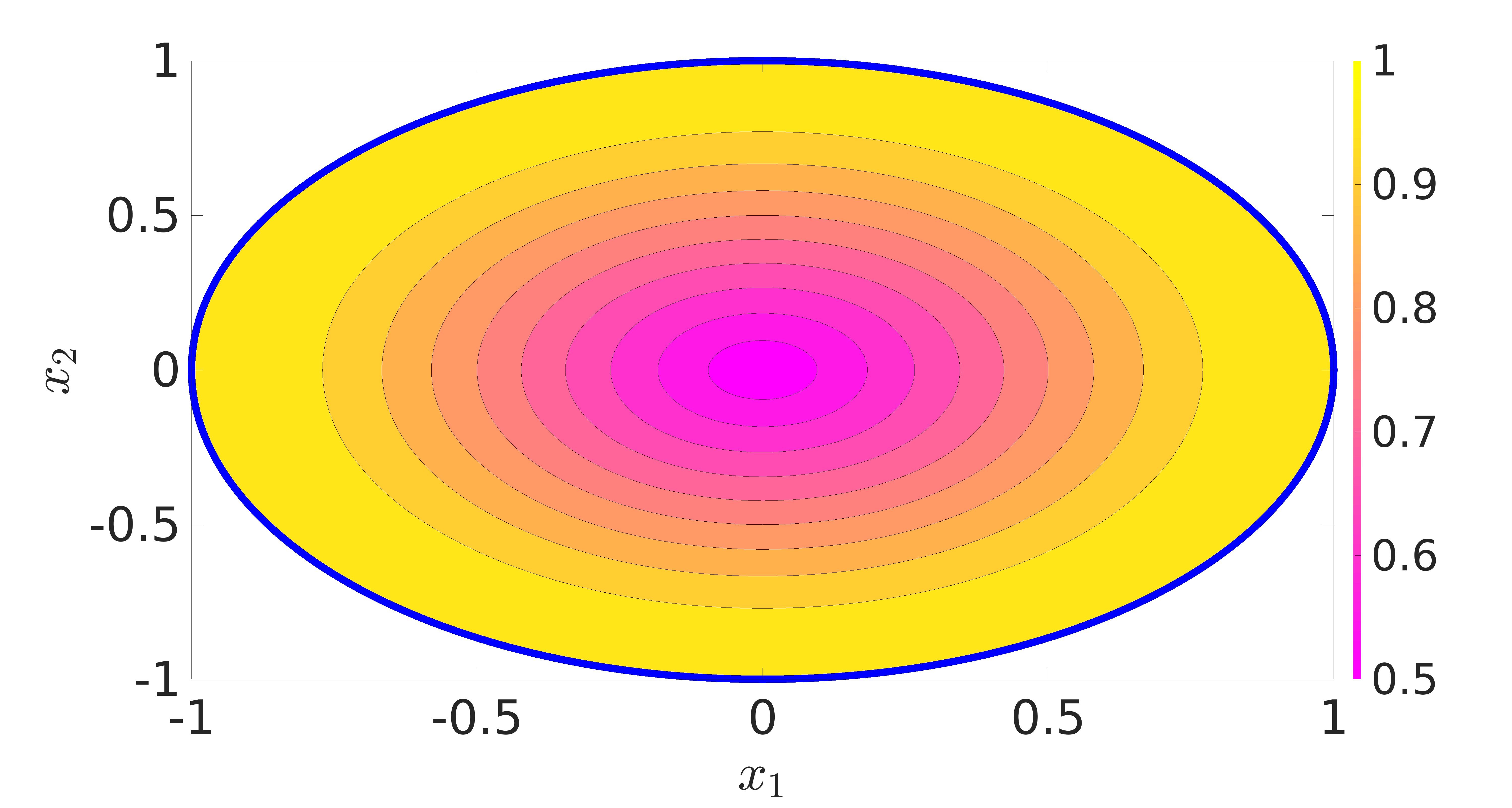}%
    \label{circle RS-SP k}}
    \subfloat[Plotting of $k_2$ for $S_2(x_3)$ (left) and $k_3$ for $S_3(x_4)$ (right).]{\includegraphics[width=3.5in]{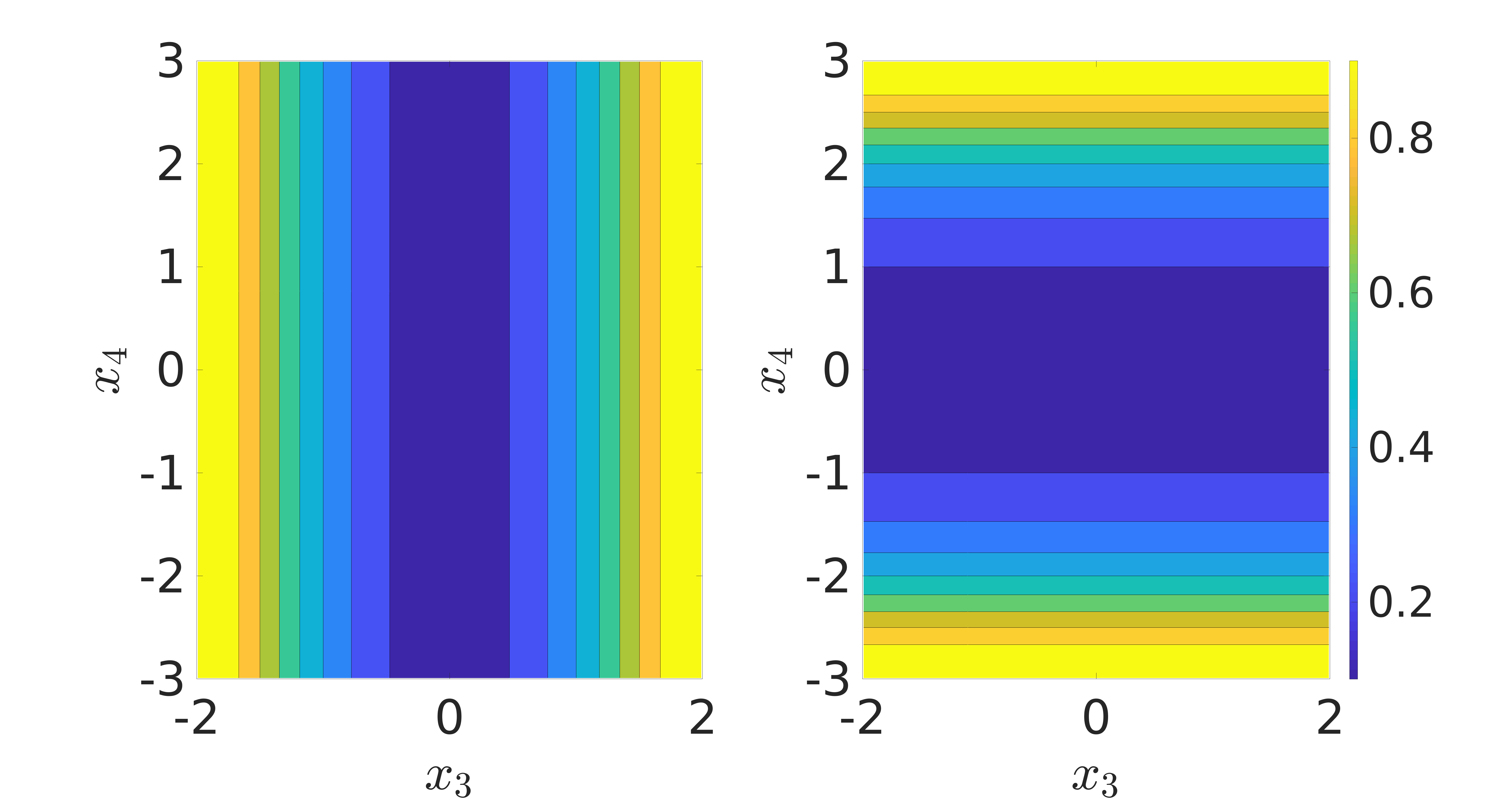}%
    \label{rectangle RS-SP k}}
    \caption{Graphical illustration of the working scheme of the RS-SP terms $S_1(x_1,x_2)$, $S_2(x_3)$,  $S_3(x_4)$ and their corresponding risk sensivity parameters $k_1$, $k_2$ and $k_3$.}
    \label{state penalty comparision}
    \end{figure*}
    
    \subsection{HJB equation for the OCP}
    Aiming at the transformed OCP in Problem \ref{Optimal control problem}, for any admissible control policies $\mu,\nu \in \Psi (\Omega)$ where $\Psi (\Omega)$ is the admissible control set \cite[Definition 1]{abu2005nearly},
    the associated optimal cost function follows
    \begin{equation} \label{optimal cost function}
    V^{*}(x(t)) = \min_{\mu,\nu \in \Psi (\Omega)}\int_{t}^{\infty} r(x(\tau),\mu(x(\tau)),\nu(x(\tau)))\,d\tau.
    \end{equation}
    Then, the HJB equation follows
    \begin{equation} \label{HJB equation}
    \begin{aligned}
   & H(x,\mu^{*}(x),\nu^{*}(x),\nabla V^{*}) = r(x,\mu^{*}(x),\nu^{*}(x) ) \\
    &+ \nabla V^{{*}^{T}}(f(x)+g(x)\mu^{*}(x)+h(x)\nu^{*}(x)) =0,
    \end{aligned}
    \end{equation}
    where the operator $\nabla$ denotes the partial derivative with regard to $x$, i.e., $\partial / \partial x $. 
    
    Assuming that the minimum of \eqref{optimal cost function} exits and is unique \cite{vamvoudakis2010online}. Based on the HJB equation \eqref{HJB equation}, the closed forms of optimal control policies $\mu^*(x)$ and $\nu^*(x)$ are obtained as \cite{abu2005nearly}
    \begin{equation} \label{optimal u}
    \mu^{*}(x) = - \beta \tanh(\frac{1}{2\beta}R^{-1}g^{\top}(x)\nabla V^{*}), 
    \end{equation}
    \begin{equation} \label{optimal v}
    \nu^{*}(x) = - \frac{1}{2\rho} h^{\top}(x)\nabla V^{*}.
    \end{equation}
    \subsection{Equivalence between Problem \ref{Robust constraint control problem} and Problem \ref{Optimal control problem}}\label{proof of problem equivalence}
    Here we defer a detailed explanation of the method to get the optimal control policies \eqref{optimal u} and \eqref{optimal v} in Section \ref{Section 3}, and focus now on the proof of equivalence between Problem \ref{Robust constraint control problem} in Section \ref{section PF} and Problem \ref{Optimal control problem} in Section \ref{section PT}.
    Comparing with the result provided in \cite{lin1998optimal} that merely considers additive disturbances, as shown in Theorem \ref{theorem for equivalence}, the additional consideration of input and state constraints further complicates the theoretical analysis.
     \begin{theorem}\label{theorem for equivalence}
     Consider the system described by \eqref{original sys} and controlled by the optimal control policy \eqref{optimal u}. 
     Suppose Assumptions \ref{bound of fg}-\ref{bound of decomposed disturbance part} hold and the initial states and control inputs lie in the predefined constraint satisfying sets \eqref{input saturation} and \eqref{state constriant}. The optimal control policy \eqref{optimal u} guarantees robust stabilization of the system \eqref{original sys} without violating input constraint \eqref{input saturation} and state constraint \eqref{state constriant}, if there exists a scalar $\epsilon_{s} \in \mathbb{R}^{+}$ such that the following inequality is satisfied
     \begin{equation} \label{equalivance condition}
    \mathcal{L}(x) > 2\rho {\nu^*}^{\top}(x)\nu^{*}(x)+\epsilon_{s}.
    \end{equation}
    \end{theorem}
    \begin{proof}
    See Appendix \ref{Theorem 1}.
    \end{proof}
    It has been proven in Theorem \ref{theorem for equivalence} that the CRSP in Problem \ref{Robust constraint control problem} is equivalent to the OCP in Problem \ref{Optimal control problem} under the inequality \eqref{equalivance condition}. 
    Thus, in order to solve the CRSP, the current task is to obtain the optimal control law \eqref{optimal u} by focusing on the OCP, which is clarified in details in the next section.
    
    \section{Approximate solutions to the OCP} \label{Section 3} 
    To get approximated solutions to the OCP, instead of introducing a common actor-critic structure used in \cite{abu2005nearly,vamvoudakis2010online}, here we adopt a single critic structure which enjoys lower computation complexity \cite{padhi2006single}. Furthermore, unlike common methods that add external noises into inputs to satisfy the required PE condition for the critic NN weight convergence \cite{vamvoudakis2010online,bhasin2013novel}, here we exploit experience data to construct the off-policy weight update law to achieve a sufficient excitation required for the critic NN weight convergence.
    Additionally, an online PER algorithm and an offline experience buffer construction algorithm are proposed as principled ways to provide the aforementioned rich enough experience data.
    
    \subsection{Value function approximation}
    According to the Weierstrass high-order approximation theorem \cite{finlayson2013method}, there exists a weighting matrix $W^{*}\in  \mathbb{R}^{N}$ such that the value function is approximated as  
    \begin{equation}\label{optimal V approximation}
    V^{*}(x) = {W^*}^{\top} \Phi(x) + \epsilon(x),
    \end{equation}
    where $\Phi(x) :  \mathbb{R}^{n} \to \mathbb{R}^{N}$ is the NN activation function, and $\epsilon(x) \in \mathbb{R}$ is the approximation error. $N$ is the number of NN activation functions. The derivative of  $V^{*}(x)$ follows 
    \begin{equation}\label{optimal dV approximation}
    \nabla V^{*}(x) = \nabla \Phi^{\top}(x)W^{*}+\nabla\epsilon(x),
    \end{equation}
    where $\nabla \Phi \in \mathbb{R}^{N \times n}$ and $\nabla\epsilon(x) \in \mathbb{R}^{n}$ are partial derivatives of $\Phi(x)$ and $\epsilon(x)$, respectively.
     As $N \to \infty$, both $\epsilon(x)$ and $\nabla\epsilon(x)$ converge to zero uniformly. Without loss of generality, the following assumption is given.
     \begin{assumption} \cite{vamvoudakis2010online}\label{bound of NN issues}
        There exist known constants $b_{\epsilon}, b_{\epsilon x}, b_\Phi, b_{\Phi x}\in \mathbb{R}^{+}$ such that $\left\| \epsilon(x)  \right\| \leq b_{\epsilon}$, $\left\| \nabla\epsilon(x)  \right\| \leq b_{\epsilon x}$, $\left\| \Phi(x) \right\| \leq b_\Phi$ and $\left\| \nabla\Phi(x) \right\| \leq b_{\Phi x}$.
     \end{assumption}
    For fixed admissible control policies $u(x)$ and $v(x)$, inserting \eqref{optimal dV approximation} into \eqref{HJB equation} yields the Lyapunov equation (LE)
    \begin{equation}\label{approximation Lyapunov equation}
        {W^*}^{\top}\nabla \Phi(f(x)+g(x)u(x)+h(x)v(x))+r(x,u(x),v(x)) = \epsilon_{h},
    \end{equation}
    where the residual error follows $\epsilon_{h} = -(\nabla \epsilon(x) )^{\top}(f(x)+g(x)u(x)+h(x)v(x)) \in \mathbb{R}$. 
    According to Assumption \ref{bound of fg}, the system dynamics is Lipschitz. This leads to the boundness of the residual error, i.e., there exists a positive scalar $b_{\epsilon_{h}}$ such that $\left\| \epsilon_{h} \right\| \leq b_{\epsilon_{h}}$.
    
    By focusing on the NN parameterized LE \eqref{approximation Lyapunov equation}, 
    unlike the common analysis and derivation process in typical ADP related works \cite{liu2020editorial,lewis2009reinforcement},
    here we rewrite \eqref{approximation Lyapunov equation} into a linear in parameter (LIP) form that follows
    \begin{equation}\label{LIP Lyapunov equation}  
    \Theta = -{W^*}^{\top}Y+\epsilon_{h},
    \end{equation}
    where $\Theta = r(x,u(x),v(x)) \in \mathbb{R}$, and $Y = \nabla \Phi(f(x)+g(x)u(x)+h(x)v(x)) \in \mathbb{R}^{N}$. Note that both $\Theta$ and $Y$ can be obtained from real-time data.

    Given the LIP form and the measurable $Y$, $\Theta$ in \eqref{LIP Lyapunov equation}, from the perspective of adaptive control, we transform the learning of critic NN weight $W^{*}$ into a parameter estimation problem of an LIP system, where $Y$ and $W^{*}$ can be treated as the regressor matrix and the unknown parameter vector of a LIP system, respectively. 
    This novel transformation enables us to design a simple weight update law with weight convergence guarantee in Section \ref{OPRL}. 
    \subsection{Off-policy weight update law} \label{OPRL}
    The ideal critic weight $W^{*}$ in \eqref{LIP Lyapunov equation} is approximated by an estimated weight  $\hat{W}$ which satisfies the following relation
    \begin{equation}\label{approximation LIP Lyapunov equation}
    \hat{\Theta} = -\hat{W}^{\top}Y,
    \end{equation}
    where $\hat{\Theta}$ is the estimated utility function.
    Denoting the weight estimation error as $\Tilde{W} = \hat{W}- W^{*}$. Then, we can get
    \begin{equation}\label{approximation error}
    \Tilde{\Theta} = \Theta - \hat{\Theta} = \Tilde{W}^{\top}Y+\epsilon_{h}.
    \end{equation}
    To achieve $\hat{W} \to W^{*}$ and $\Tilde{\Theta} \to \epsilon_{h}$,  $\hat{W}$ should be updated to minimize $E = \frac{1}{2} \Tilde{\Theta}^{\top}\Tilde{\Theta}$.
    Furthermore, in order to guarantee the weight convergence while minimizing $E$, rather than incorporating external noises to satisfy the PE condition \cite{vamvoudakis2010online,bhasin2013novel}, here we exploit experience data to accelerate the online learning process, which could achieve the sufficient excitation required for the weight convergence. Finally, we design a simple yet efficient off-policy weight update law of the critic NN that follows
    \begin{equation} \label{w update law}
        \dot{\hat{W}} = - \Gamma k_c Y\Tilde{\Theta} -  \sum_{l=1}^{P} \Gamma k_{e} Y_l\Tilde{\Theta}_{l},
    \end{equation}
    where $\Gamma \in \mathbb{R}^{N \times N}$ is a constant positive definite gain matrix. $k_c, k_{e} \in \mathbb{R}^{+}$ are constant gains to balance the relative importance between current and experience data to the online learning process. $P \in \mathbb{N}^{+}$ is the volume of the experience buffers $\mathfrak{B}$ and $\mathfrak{E}$, i.e., the maximum number of recorded data points. The regressor matrix $Y_{l}\in \mathbb{R}^{N}$ and the approximation error $\Tilde{\Theta}_{l}\in \mathbb{R}$ denote the $l$-th collected data of the corresponding experience buffers $\mathfrak{B}$ and $\mathfrak{E}$, respectively. The developed critic NN weight update law \eqref{w update law} is in a different form comparing with the counterpart in typical ADP related works (See \cite{liu2020editorial,lewis2009reinforcement,vamvoudakis2010online} and the references therein). Our proposed weight update law \eqref{w update law} is easily implemented and enjoys guaranteed weight convergence without causing undesirable oscillations and additional control effort expenditures.
    
    To analyse the weight convergence of the critic NN, a rank condition about the experience buffer $\mathfrak{B}$, which serves as a richness criterion of the recorded experience data, is firstly clarified here in Assumption~\ref{rank condition}.
    \begin{assumption} \label{rank condition}
          Given an experience buffer $\mathfrak{B} = [Y^{\top}_{1},...,Y^{\top}_{P}] \in \mathbb{R}^{N \times P}$, 
          there holds $rank(\mathfrak{B}) = N$.
    \end{assumption}
    Comparing with the traditional PE condition given in \cite{tao2003adaptive}, the rank condition regarding $\mathfrak{B}$ in Assumption~\ref{rank condition} provides an index about the data richness that can be checked online, which is favourable to controller designers.

    Based on the aforementioned settings, the critic NN weight convergence proof is shown as follows.
    \begin{theorem} \label{theorem of weight convergency}
        Given Assumption~\ref{rank condition}, the weight learning error $\Tilde{W}$ converge to a small neighbourhood around zero.
    \end{theorem}
    \begin{proof} \label{proof of weight convergency}
        Consider the following candidate Lyapunov function
        \begin{equation} \label{Vcl}
            V_{er} = \frac{1}{2} \Tilde{W}^{\top} \Gamma^{-1} \Tilde{W}. 
        \end{equation}
        The time derivative of $V_{er}$ reads
        \begin{equation} \label{dVcl}
            \begin{aligned}
                \dot{V}_{er} &= \Tilde{W}^{\top} \Gamma^{-1}(- \Gamma k_c Y\Tilde{\Theta} - \Gamma  \sum_{l=1}^{P}k_{e}Y_{l}\Tilde{\Theta}_{l})\\
                &= -k_c \Tilde{W}^{\top}Y\Tilde{\Theta} -\Tilde{W}^{\top} \sum_{l=1}^{P}k_{e} Y_{l}\Tilde{\Theta}_{l}\\
                &\leq - \Tilde{W}^{\top} B \Tilde{W} +\Tilde{W}^{\top}\epsilon_{er}.\\
            \end{aligned}
        \end{equation}
        where $B := \sum_{l=1}^{P} k_{e}Y_{l}Y^{\top}_l$, and $\epsilon_{er} :=-k_c Y\epsilon_{h}-\sum_{l=1}^{P}k_{e} Y_{l}\epsilon_{h_{l}}$. The boundness of $Y$ and $\epsilon_{h}$ results in bounded $\epsilon_{er}$, i.e., there exists $b_{\epsilon_{er}} \in \mathbb{R}^{+}$ such that $\left\|\epsilon_{er}\right\| \leq b_{\epsilon_{er}}$. Since $B$ is positive definite according to Assumption~\ref{rank condition}, \eqref{dVcl} can be written as
            \begin{equation} \label{dVcl 1}
            \begin{aligned}
                \dot{V}_{er} 
                & \leq -\left\| \Tilde{W} \right\| (\lambda_{\min}(B)\left\| \Tilde{W} \right\|- b_{\epsilon_{er}}).
            \end{aligned}
        \end{equation}
        Therefore, $\dot{V}_{er} < 0$ if $\left\| \Tilde{W} \right\| > \frac{b_{\epsilon_{er}}}{\lambda_{\min}(B)}$. Finally, it is concluded that the weight estimation error of the critic NN will converge to the residual set $\Omega_{\Tilde{W}}$ defined as
        \begin{equation} \label{compacrt set TildeW}
            \Omega_{\Tilde{W}} = \left\{\Tilde{W}| \left\| \Tilde{W} \right\| \leq \frac{b_{\epsilon_{er}}}{\lambda_{\min}(B)}\right\}.
        \end{equation}
        This completes the proof. 
    \end{proof}
    By observing \eqref{compacrt set TildeW}, the size of $\Omega_{\Tilde{W}}$ relates with the bound of $\epsilon_{er}$. As $N \to \infty$, we know that $\epsilon_{h} \to 0$ results in $\epsilon_{er} \to 0$. Then, we get $\dot{V}_{er} \leq -\lambda_{\min}(B) \left\| \Tilde{W} \right\|^2 $, i.e., $\Tilde{W} \to 0$ exponentially as $t \to \infty$.
    Equivalently, it is guaranteed that $\hat{W}$ converges to $W^{*}$. 
    Finally, in conjugation with~\eqref{optimal u} and~\eqref{optimal v}, the approximate optimal control policies are obtained as
    \begin{equation} \label{approximation u}
    \hat{u}(x) = - \beta \tanh(\frac{1}{2\beta}R^{-1}g^{\top}(x)\nabla \Phi^{\top}(x)\hat{W}), 
    \end{equation}
    \begin{equation} \label{approximation v}
    \hat{v}(x) = - \frac{1}{2\rho} h^{\top}(x)\nabla \Phi^{\top}(x)\hat{W}.
    \end{equation}

    In the following part, the main conclusions are provided based on the off-policy weight update law \eqref{w update law} and the approximate optimal control policies \eqref{approximation u}, \eqref{approximation v}.
    \begin{theorem} \label{final theorem for optimal problem}
    Consider the dynamics \eqref{auxiliary system}, the off-policy weight update law of the critic NN in \eqref{w update law}, and the control policies \eqref{approximation u} and \eqref{approximation v}. Given Assumptions \ref{bound of fg}-\ref{rank condition}, for sufficiently large $N$, the approximate control policies \eqref{approximation u} and \eqref{approximation v} stabilize the system \eqref{auxiliary system}. Moreover, the critic NN weight learning error $\Tilde{W}$ is uniformly ultimately bounded (UUB).
    \end{theorem}
    \begin{proof}
    See Appendix \ref{Theorem 3}.
    \end{proof}
     Assumption \ref{rank condition} in Theorem \ref{final theorem for optimal problem} is the prerequisite to ensure that $\hat{W}$ converges to $W^{*}$. The guaranteed weight convergence enables us to directly apply $\hat{W}$ in \eqref{w update law} to construct the approximate optimal control policies \eqref{approximation u}, \eqref{approximation v}. Assumption \ref{rank condition} is not restrictive and can be easily satisfied by the algorithms proposed in the next section.
     \subsection{Experience replay with online and offline experience buffers} \label{section PER}
     To get rich enough experience data to satisfy Assumption \ref{rank condition},
     given the sampling deficiency problem of the subsequent way of data usage in existing ADP related works \cite{yang2019online,kamalapurkar2016model}, 
     and inspired by the concurrent learning (CL) technique developed for system identification \cite{chowdhary2010concurrent},
     here we design both online and offline principled methods to provide the required rich enough experience data to satisfy Assumption \ref{rank condition}. To the best of our knowledge, this is the first attempt to design principled ways of replaying rich enough experience data to achieve the required exploration for the critic NN weight convergence.
     \subsubsection{Online PER algorithm}
     Before the estimated weight converges (line 4-5), Algorithm \ref{DataSelection} chooses the minimum eigenvalue (i.e., $\lambda_{\min}(\mathfrak{B})$) as the priority scheme to filter experience data $Y$, $\Theta$ recorded into the experience buffers $\mathfrak{B}$ and $\mathfrak{E}$, respectively. 
     Here the prioritized criterion is different from ones used in existing PER algorithms \cite{schaul2015prioritized}.
     We prefer experience data accompanied with a larger $\lambda_{\min}(\mathfrak{B})$ given the facts that: a) a nonzero $\lambda_{\min}(\mathfrak{B})$ ensures that $rank(\mathfrak{B}) = N$ in Assumption \ref{rank condition} holds \cite{chowdhary2010concurrent}, i.e., the convergence of $\hat{W}$ to $W^*$ is guaranteed; b) according to \eqref{dVcl 1} and \eqref{compacrt set TildeW}, a larger $\lambda_{\min}(\mathfrak{B})$ leads to a faster weight convergence rate and a smaller residual set.
     Although efficient, the priority scheme $\lambda_{\min}(\mathfrak{B})$ accompanies with additional computation loads.
     Thus, on convergence (line 6-7), i.e., we have gotten sufficient excitation required for the weight convergence, we alternate to a low-cost way where recent experience data are sequentially recorded.
     This sequent way of data usage enjoys robustness to a dynamic environment since collected real-time data could reflect environmental changes in time.
     Unlike standard methods that first construct a huge experience buffer and then sample partial data \cite{fedus2020revisiting}, to reduce computation loads and relieve hardware requirements, we directly build experience buffers with a limited buffer size $P$ here, and all of the recorded experience data are replayed to the critic NN for the online weight learning. 
     The buffer size $P$ is a hyper-parameter that requires careful tuning. In order to satisfy Assumption \ref{rank condition}, $P$ is selected such that $P \geq N$ holds.  
     \subsubsection{Offline experience buffer construction algorithm}
     In addition to experience buffers $\mathfrak{B}$ and $\mathfrak{E}$ constructed from online data featured with sufficient richness and high minimum eigenvalues, in Algorithm \ref{offline DataSelection}, we construct experience buffers $\mathcal{F},\mathcal{G},\mathcal{H},\mathcal{K},\mathcal{R} \in \mathbb{R}^{N \times P}$ with offline collected experience data. Here the offline experience data are generated from the pre-simulation within the given operation region $A$.
     To generate rich enough training data, for $i$-th dimension of an allowable operation region $A_i \in \mathbb{R}$, we can sample training data equidistantly with a fixed mesh size $\delta_i \in \mathbb{R}^{+}$, or a predefined number $c_i \in \mathbb{N}^{+}$. 
     Besides, rather than sampling partial data from the offline constructed experience buffers based on an uniform or a prioritized way \cite{fedus2020revisiting}, we replay all offline experience data in an average way for the weight estimation. Thereby, the off-policy weight update law \eqref{w update law} based on Algorithm \ref{offline DataSelection} is redesigned as 
    \begin{equation} \label{w update law offline}
        \dot{\hat{W}} = - \Gamma k_c Y\Tilde{\Theta} - \frac{1}{P}\sum_{l=1}^{P} \Gamma k_{e} Y_l\Tilde{\Theta}_{l}.
    \end{equation}
    The implementation of using offline recorded experience data to accelerate the online learning process enjoys two advantages: the rank condition in Assumption \ref{rank condition} is easily satisfied, and the possible influence of data noises is offset by averaging. 
    \begin{remark}
        The subsequent \cite{yang2019online,kamalapurkar2016model} or uniform sampling way \cite{fedus2020revisiting} of data usage cannot ensure that Assumption \ref{rank condition} is always satisfied, whereas, the developed online PER algorithm not only guarantees establishment of Assumption \ref{rank condition}, but also accelerates the weight convergence speed.
        Besides, in the CL algorithm \cite{chowdhary2010concurrent}, a constant hyperparameter that highly depend on prior knowledge is required to characterize the data richness, which results in additional parameter tuning effort. This is avoided in Algorithm \ref{DataSelection} by calculating the rank value of the experience buffer $\mathfrak{B}$, even though additional (but acceptable) computation load is unintentionally introduced.
    \end{remark}
    \begin{remark}
     The simple three-layer NNs adopted in this paper provides us with opportunities to revisit the ER technique and investigate principled ways to exploit experience data to accelerate the online learning process, which is difficult in the deep RL field because the complexity of deep NNs hinders researchers from understanding the mechanism of the ER technique \cite{fedus2020revisiting}.
     Algorithms \ref{DataSelection}-\ref{offline DataSelection} are plug-in methods, which can be widely used with little extra computations and engineering efforts. We argue that they also provide alternative solutions for the parameter convergence problem existing in the adaptive control field.
    \end{remark}
        \begin{algorithm}[!t]
        \caption{Online Prioritized Experience Replay Algorithm} 
        \label{DataSelection}
        \begin{algorithmic}[1] 
            \Require Iteration index: $n_r$;  Buffer size: $P$; Threshold: $\xi$.
            \Ensure Experience buffers: $\mathfrak{B}$, $\mathfrak{E}$.
            \If{$n_r \leq P$}
                \State Record current $Y$, $\Theta$ into $\mathfrak{B}$, $\mathfrak{E}$ respectively.
            \Else
                \If {$\left\|W_{n_{r}}-W_{n_{r}-1}\right\| > \xi$}
                    \State Record prioritized $Y$, $\Theta$ leading to high  $\lambda_{\min}(\mathfrak{B})$.
                \Else
                    \State Record current $Y$, $\Theta$ sequentially to update $\mathfrak{B}$,$\mathfrak{E}$.
                \EndIf
            \EndIf
            \end{algorithmic}
            \end{algorithm}
    \begin{algorithm}[!t]
        \caption{Offline Experience Buffer Construction Algorithm} 
        \label{offline DataSelection}
        \begin{algorithmic}[1] 
            \Require $A = [\underline{A},\overline{A}]$ with $\underline{A}, \overline{A}\in \mathbb{R}^{n}$; Mesh size : $\delta \in \mathbb{R}^{n}$, or data point number: $c \in \mathbb{R}^{n}$; Empty sets: 
            $\mathcal{X} \in \mathbb{R}^{n \times d}$.
            \Ensure $\mathcal{F}$; $\mathcal{G}$; $\mathcal{H}$; $\mathcal{K}$; $\mathcal{R} $; $P$.
            \State Sampling: $\mathcal{X} \xleftarrow[\delta]{c} A$; $P \leftarrow \prod_{i = 1}^n (\overline{A}_i-\underline{A}_i)/\delta_i$, or $\prod_{i = 1}^n c_i$.
            \State Data collection: $\mathcal{F} \leftarrow \nabla \Phi^{\top}(\mathcal{X})f(\mathcal{X})$; $\mathcal{R} \leftarrow r_d(\mathcal{X})$\\
             $\mathcal{G} \leftarrow \nabla \Phi^{\top}(\mathcal{X})g(\mathcal{X})$;
             $\mathcal{H}\leftarrow \nabla \Phi^{\top}(\mathcal{X})h(\mathcal{X})$; $\mathcal{K} \leftarrow \mathcal{L} (\mathcal{X})$
            \end{algorithmic}
        \end{algorithm}    
\section{Numerical simulations} \label{Section 4}
    This section provides three simulation examples to show the effectiveness of the proposed off-policy weight update laws \eqref{w update law}, \eqref{w update law offline}, the approximate optimal control policy \eqref{approximation u}, and Algorithms \ref{DataSelection}-\ref{offline DataSelection}.
    Firstly, we consider an optimal regulation problem (ORP) of a nonlinear system \cite{vamvoudakis2010online} in Section \ref{second-order nonlinear system simulation case}. This ORP serves as a benchmark problem to prove that both Algorithm \ref{DataSelection} based \eqref{w update law} and  Algorithm \ref{offline DataSelection} based \eqref{w update law offline} enable the estimated critic NN weight to converge to the actual optimal value, which is marginally considered in existing single critic structure related works \cite{heydari2012finite,yang2018event}.
    Then, to show the superiority of our proposed off-policy risk-sensitive RL-based control framework to counter input and state constraints under additive disturbances, a 
    pendulum system \cite{liu2013policy} is investigated in Section \ref{simulation pendulum}. 
    Furthermore, a target (set-point) tracking control problem of a 2-DoF robot manipulator \cite{li2021concurrent} is investigated in Section \ref{simulation 2-DoF} to validate the real-time performance of our proposed control framework.
    \subsection{Example 1: ORP of a benchmark nonlinear system}\label{second-order nonlinear system simulation case}
    To validate that based on our proposed off-policy weight update laws \eqref{w update law}, \eqref{w update law offline}, and Algorithms \ref{DataSelection}-\ref{offline DataSelection}, the estimated weight guarantees convergence to the actual optimal value, a benchmark problem \cite{vamvoudakis2010online} is investigated here. Note that only an optimal regulation problem without considering disturbances nor input/state constraints is investigated here. Otherwise, the actual optimal value of the NN weight is unknown. The benchmark continuous-time nonlinear system is given as
        \begin{equation*}
        \dot{x} = f(x) + g(x)u, \ \  x \in \mathbb{R}^{2},
    \end{equation*}
    where $\small{f(x)= \begin{bmatrix}
            -x_{1}+x_{2},
            -0.5x_{1}-0.5x_{2}(1-(\cos2x_{1}+2)^2)
        \end{bmatrix} ^{\top}}$,
        $ g(x)= \begin{bmatrix}
            0 ,
            \cos(2x_{1})+2 
        \end{bmatrix}^{\top}$.
    The standard quadratic cost function follows
    \begin{equation*}\label{second-order nonlinear system cost function}
        V(x) = \int_{0}^{\infty} x^{\top}Qx  + u^{\top}Ru \,dt,
    \end{equation*}
    where $Q = I_{2 \times 2}$ and $R = 1$. By following the converse HJB approach developed in \cite{nevistic1996constrained}, we get the optimal value function as $V^{*} = 0.5 x^2_1 + x^2_2$. Thus, by choosing the activation function as $\Phi(x) = [x^{2}_{1},x_{1}x_{2},x^{2}_{2}]^{\top}$, the optimal weight follows $W^* = [0.5,0,1]^{\top}$. Initial values are set as $x(0) = [1,1]^{\top}$, $\hat{u}(0) =0$.
    For the off-policy weight update law \eqref{w update law} based on Algorithm \ref{DataSelection}, we choose $P = 5$, $k_c = 1$, $k_e = 1$, $\Gamma = \diag{2,1.4,1}$,  and $\xi = 10^{-3}$. It is displayed in Fig.\ref{online weight case} that after $1$ s, the estimated weight converges to
    \begin{equation*} \label{nonlinear system esitmated weight online case}
        \hat{W}_1 = [0.5040,0.0592,1.0625]^{\top}.
    \end{equation*}
    Regarding the weight update law \eqref{w update law offline} under Algorithm \ref{offline DataSelection}, we start with constructing offline experience buffers by sampling $10$ data points separately for $x_1 \in A_1 =[-2,2]$, $x_2 \in A_2 = [-4,4]$.
    $P = 100$, $k_c = 1$, $k_e = 1$, and $\Gamma =  \diag{5,0.5,0.01}$ are chosen during the online operation. As displayed in Fig.\ref{offline weight case}, the estimated weight converges to
        \begin{equation*} \label{nonlinear system esitmated weight offline case}
        \hat{W}_2 = [0.50721,-0.0417,0.9783]^{\top}.
    \end{equation*}
    Thus, it is concluded that the weight update laws \eqref{w update law}, \eqref{w update law offline} under 
    Algorithms \ref{DataSelection}-\ref{offline DataSelection} ensure that $\hat{W}$ converges to $W^*$.
     Comparing with the results shown in \cite{vamvoudakis2010online}, without incorporating external noises to satisfy the PE condition, our proposed off-policy weight update laws enable $\hat{W}$ converge to $W^*$ with a fast speed.
      \begin{figure}[!t]
    \subfloat[ $\hat{W}_1$ under online Algorithm \ref{DataSelection}.]{\includegraphics[width=1.85in]{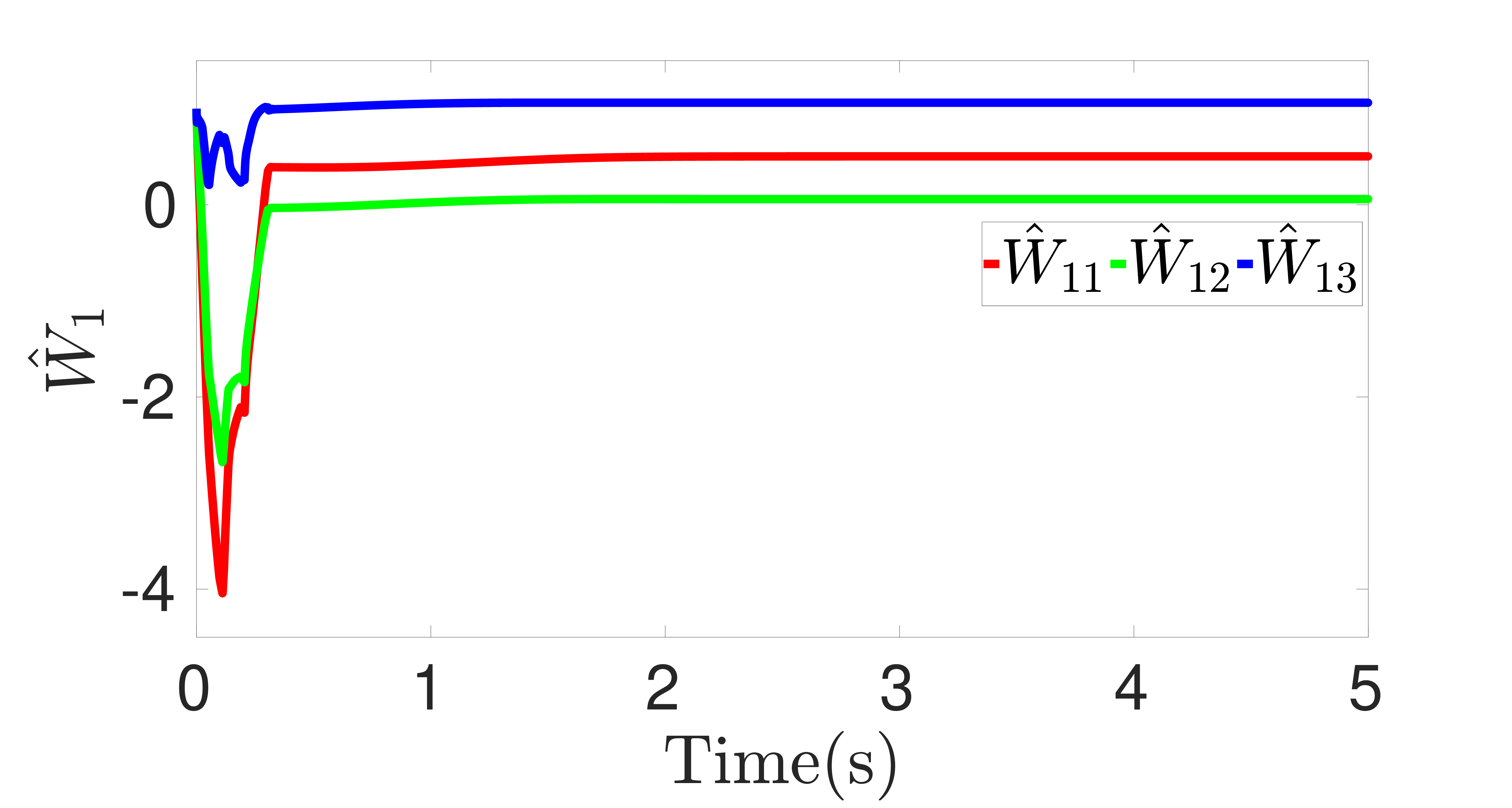}%
    \label{online weight case}}
    \subfloat[$\hat{W}_2$ under offline Algorithm \ref{offline DataSelection}.]{\includegraphics[width=1.85in]{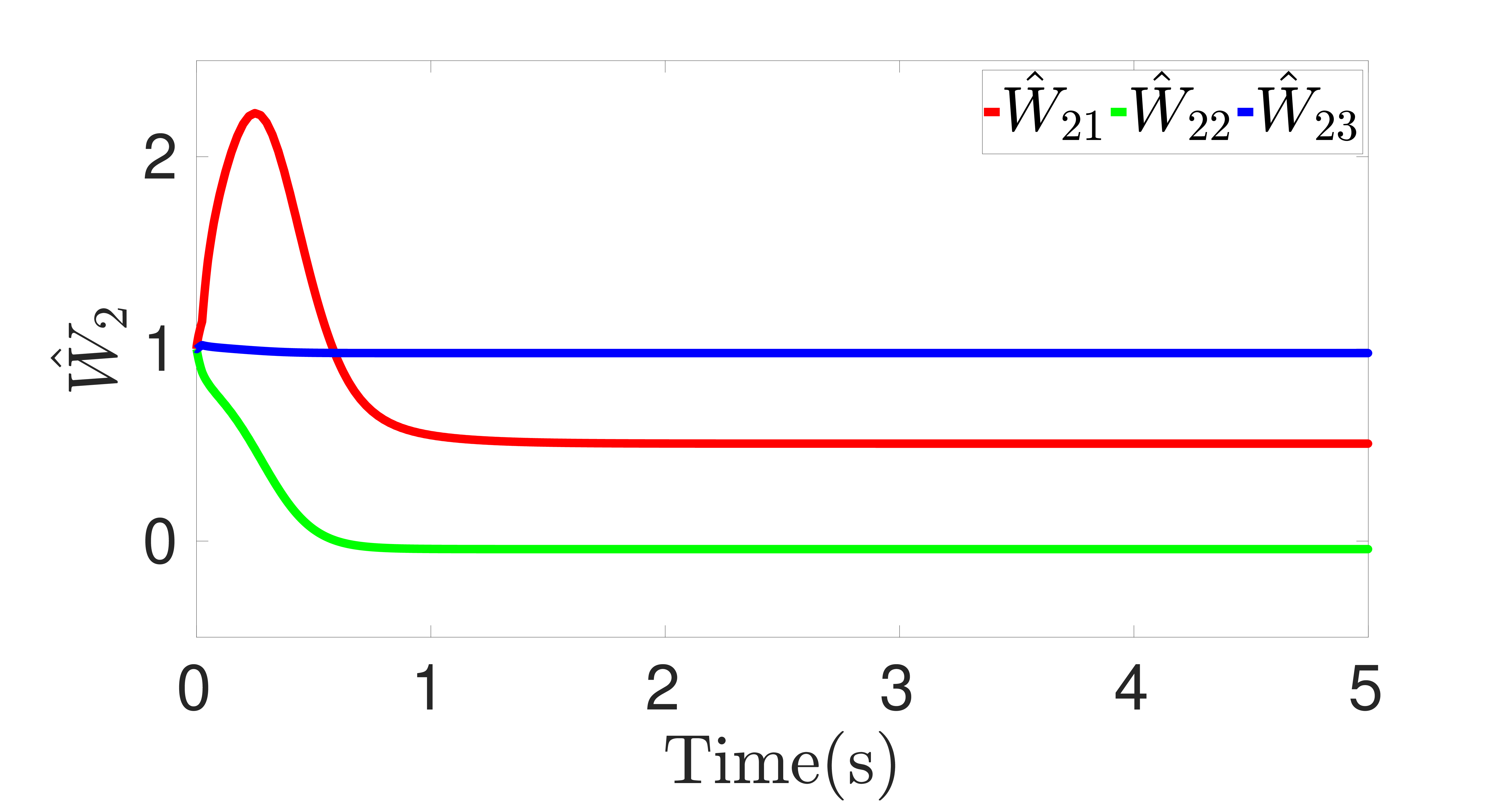}%
    \label{offline weight case}}
    \caption{The evolution trajectory of the estimated weight.}
    \label{weight convergence of benchmark problem}
    \end{figure}

\subsection{Example 2: CRSP of a pendulum system} \label{simulation pendulum}
    This section investigates the CRSP (Problem \ref{Robust constraint control problem}) of a pendulum system to validate the effectiveness and efficiency of the proposed control framework shown in Fig.\ref{algorithm framework}.
    The state-space model of a pendulum system is given as \cite{liu2013policy}
    \begin{equation*}\label{pendulum system original system}
        \dot{x} = f(x) +g(x)u +k(x)d(x),
    \end{equation*}
    where $f(x)= \begin{bmatrix}
            x_{2},
            -4.9\sin{x_1}-0.2x_2 
        \end{bmatrix}^{\top}$,
        $g(x)=\begin{bmatrix}
            0,
            0.25 
        \end{bmatrix}^{\top}$,
        $k(x)=\begin{bmatrix}
            1,
            -0.2
        \end{bmatrix}^{\top}$, and
        $d(x)=\omega_{1} x_1 \sin(\omega_{2}x_2)$. During the simulation, $\omega_{1}$ and $\omega_{2}$ are randomly chosen within the scope $[-\sqrt{2}/2,\sqrt{2}/2]$ and $[-2,2]$,
        respectively. Thus, $\left\| d(x)\right\| \leq \sqrt{2}/2 \left\| x\right\|$ and $\left\| g^{\dagger}(x)k(x)d(x)\right\| \leq 0.4\sqrt{2} \left\| x\right\|$ establishes, and Assumptions \ref{bound of d and inverse term}-\ref{bound of decomposed disturbance part} can be satisfied by choosing $d_M(x) =\sqrt{2}/2 \left\| x\right\|$ and $l_M(x)= 0.4\sqrt{2} \left\| x\right\|$. Here the input and state constraints are considered as $\mathbb{U} = \left\{u \in \mathbb{R} : |u| \leq \beta \right\}$ and $\mathbb{X}_1 = \left\{x_1 \in \mathbb{R} : h_1(x_1) = |x_{1}|-\alpha_{1}<0 \right\}$, $\mathbb{X}_2 = \left\{x_2 \in \mathbb{R} : h_2(x_2)=|x_{2}|-\alpha_{2} <0\right\}$, where $\beta = 1.5$, $\alpha_1 = 2.01$, and $\alpha_2 = 4$.
        The corresponding auxiliary system follows
        \begin{equation}\label{pendulum system auxiliary system}
        \dot{x} = f(x)+g(x)u+h(x)v,
        \end{equation}
        where $h(x) = [1,0]^{\top}$.
    Let $\rho=0.1$, $k_i = 1/(1+(\alpha_i - \left| x_i\right|)^2), i =1,2$, $Q = I_{2 \times 2}$, and $R = 1$. For the CRSP of pendulum, the cost function of \eqref{pendulum system auxiliary system} follows
    \begin{equation*}\label{utility function with constraints pendulum}
        V_{c}(x) = \int_{0}^{\infty} \mathcal{W}(u)+\mathcal{L}(x)+0.1v^{\top}v+0.37\left\| x \right\|^2\,dt,
    \end{equation*}
    where $\mathcal{W}(u) = 2\beta R u \tanh^{-1}(u/\beta)+\beta^2 R\log(1-u^2/\beta^2)$, $\mathcal{L}(x) = x^{\top}Qx+ k_1\log (\alpha^2_{1}/(\alpha^2_{1}-x^2_{1}))+k_2 \log (\alpha^2_{2}/(\alpha^2_{2}-x^2_{2}))$.
    For comparison, a robust optimization problem (ROP) for the pendulum without considering input and state constraints is also investigated here.
    Regarding the ROP case, the common quadratic cost function given in \cite{lin1992robust} is recalled  here. By setting $Q_o = I_{2\times 2}$ for the $x^{\top}Q_ox$ term and $R_o = 1$ for the $u^{\top}R_o u$ term, the cost function of \eqref{pendulum system auxiliary system} follows
    \begin{equation*}\label{utility function without constraints pendulum}
        V_{o}(x) = \int_{0}^{\infty} 1.37\left\| x \right\|^2 + u^{\top}u+0.1v^{\top}v \,dt.
    \end{equation*}
    The activation function for the pendulum is chosen as $\Phi(x) = [x^{2}_{1},x_{1}x_{2},x^{2}_{2},x^{3}_{2},x_{1}x^{2}_{2},x^{2}_{1}x_{2}]^{\top}$. 
    For the off-policy weight update law \eqref{w update law} based on Algorithm \ref{DataSelection}, $P= 9$, $k_c = 0.01$, $k_e = 0.001$, $\Gamma = I_{6 \times 6}$,  and $\xi = 10^{-3}$  are chosen. 
    The initial valuesare set as $x(0) = [2,-2]$, $\hat{u}(0) = 0$, $\hat{v}(0) = 0$, and $\hat{W} = [1,1,1,1,1,1]^{\top}$. Note that the initial values of states are purposely set to be near the constraint boundaries to highlight the efficacy of our method. Simulation results for the pendulum system are shown as follows.
    \begin{figure*}[!t]
    \centering
    \subfloat[Convergence result of $\hat{W}$ for CRSP case.]{\includegraphics[width=2.4in]{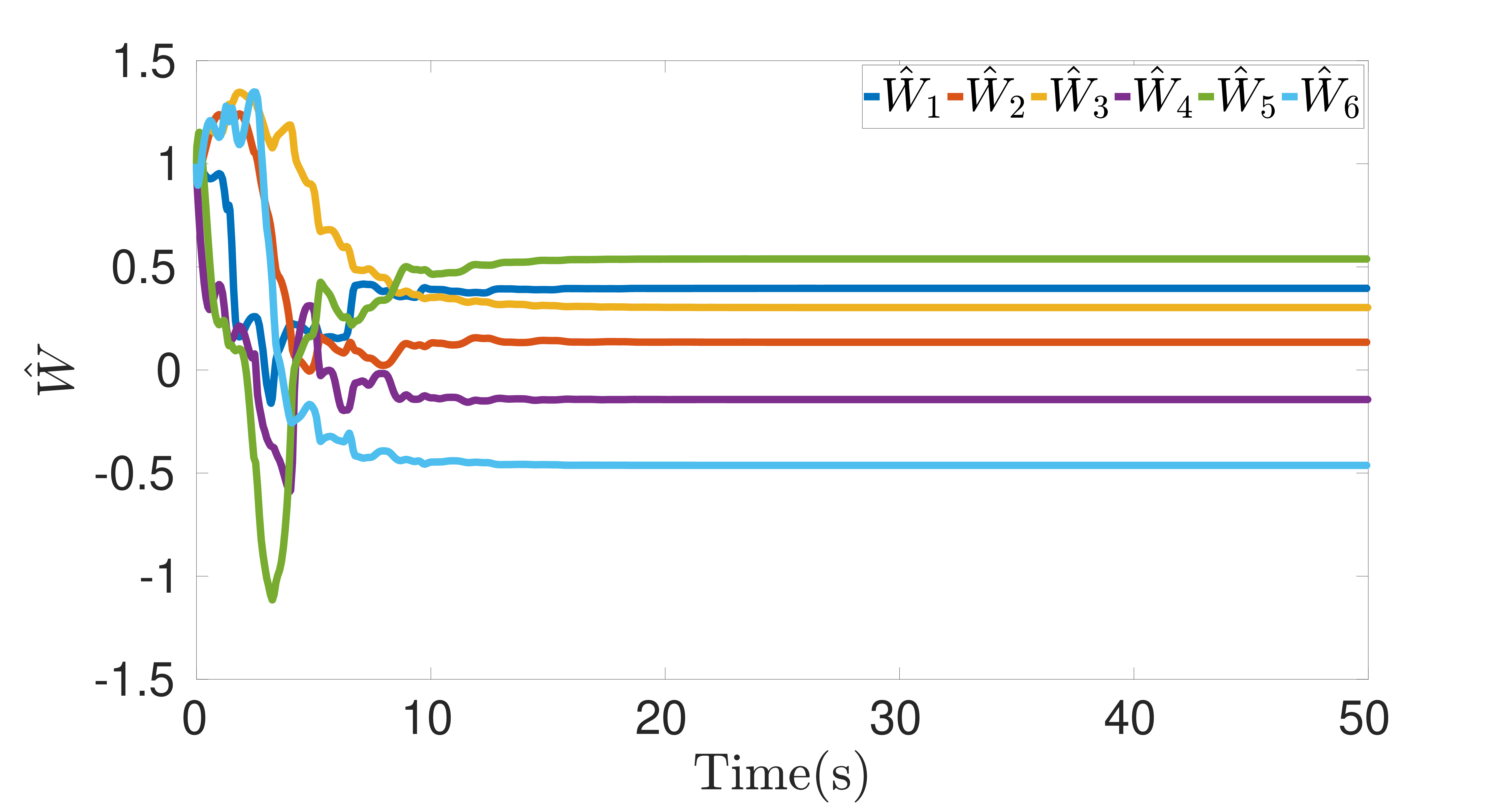}%
    \label{fig of critic weight pendulum}}
    \subfloat[Trajectories of control input $\hat{u}$ .]{\includegraphics[width=2.4in]{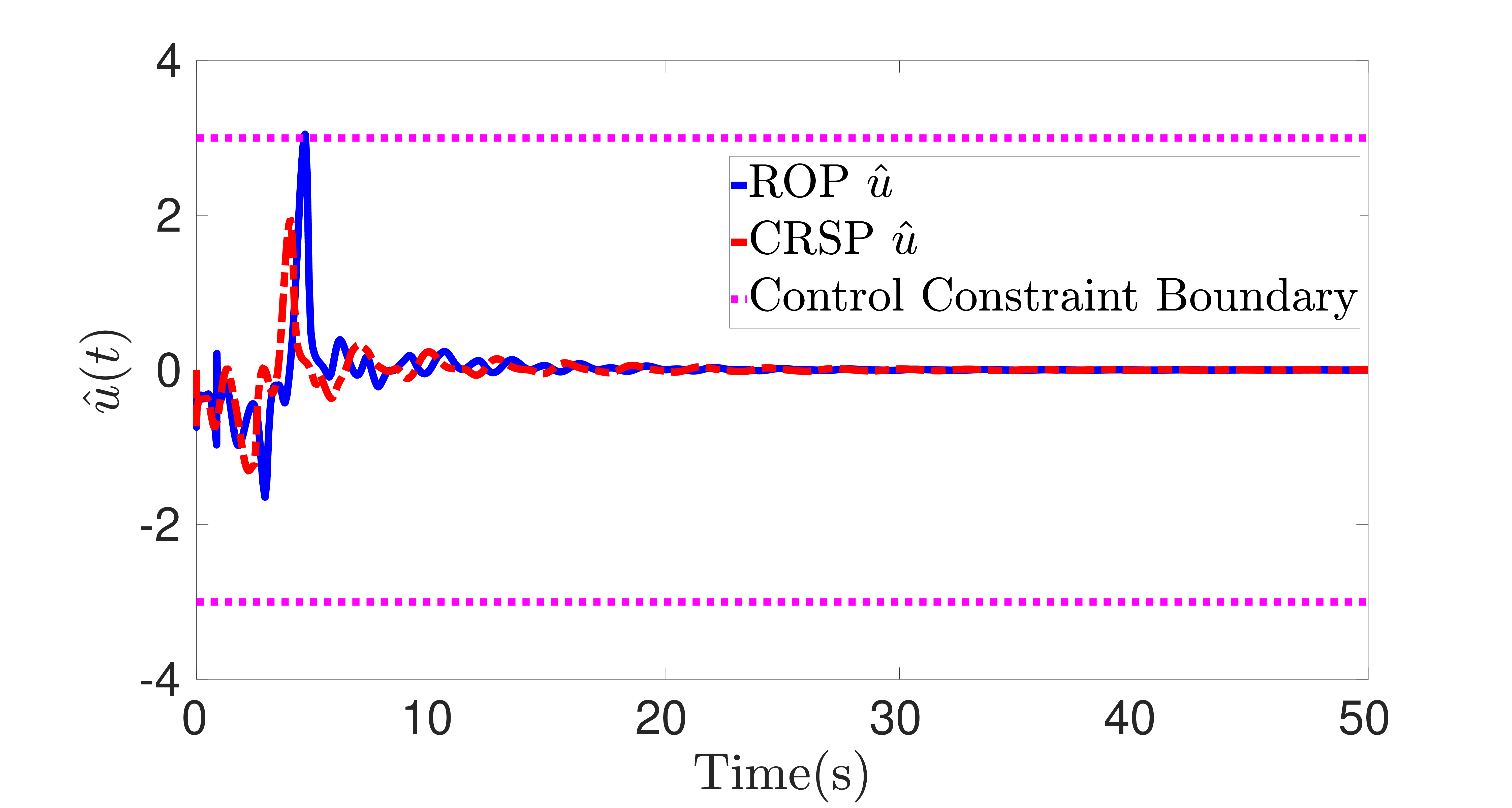}%
    \label{fig of control strategy comparision pendulum}}
    \subfloat[Phase plot of states $x$.]{\includegraphics[width=2.4in]{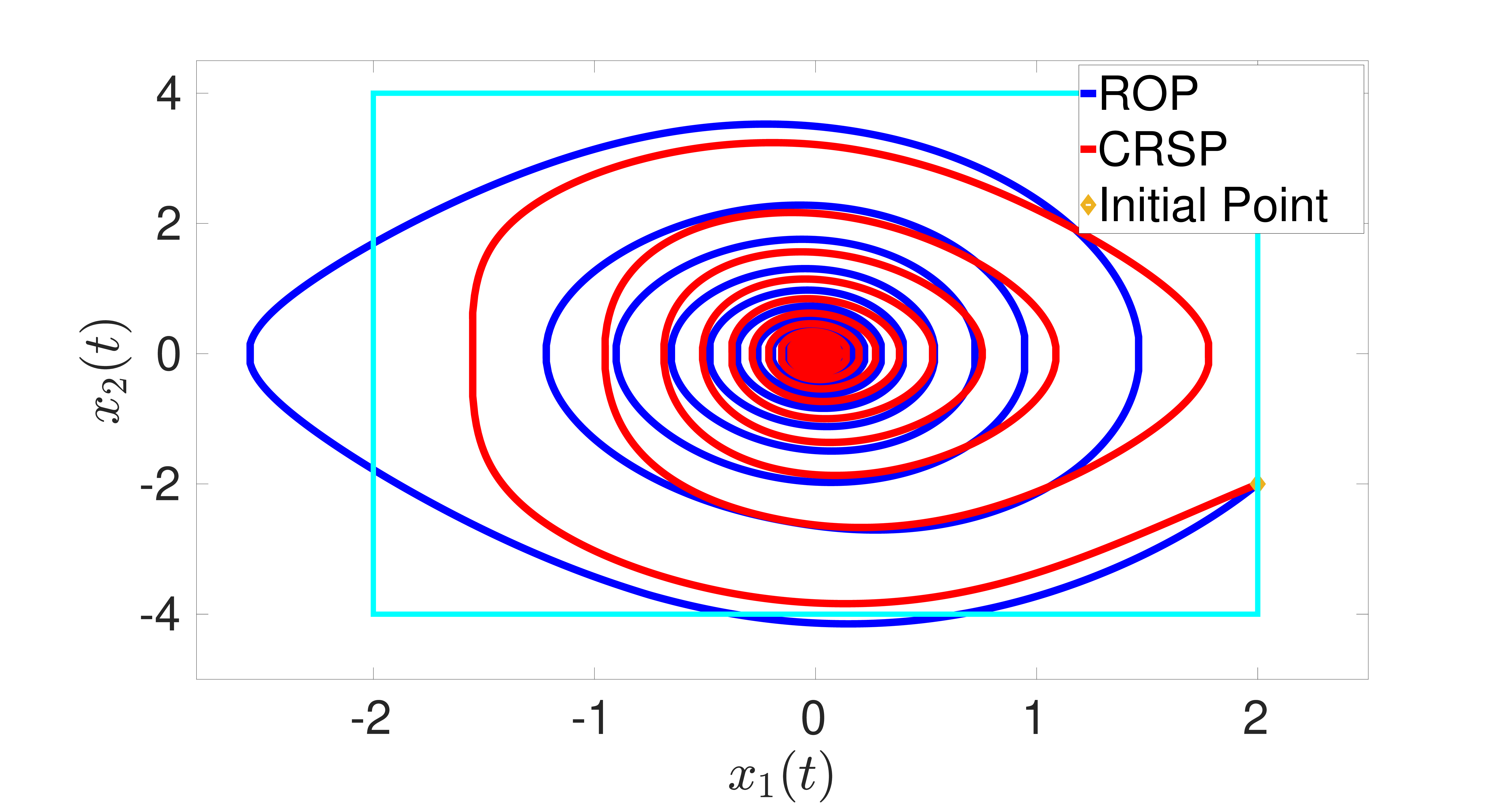}%
    \label{phase plot pendulum}}
    \caption{The comparison results for CRSP and ROP of the pendulum system.}
    \label{simulation result of pendulum}
    \end{figure*}
    
    The evolution trajectory of the estimated critic NN weight of the CRSP case is shown in Fig.\ref{fig of critic weight pendulum}. It is observed that after $ 20 \mathrm{s}$, the convergence result of $\hat{W}$ achieves.
    As displayed in Fig.\ref{fig of control strategy comparision pendulum}, the control trajectory of the CRSP case illustrates that the RS-IP term enables the satisfaction of input constraints. However, for the ROP case, the input constraint boundary is violated.
    The phase portrait of states is provided in Fig.\ref{phase plot pendulum} where the cyan rectangle represents the state constraint boundary. It is observed that states $x_1$, $x_2$ asymptotically converge to the equilibrium point and always lie in the state constraint set for the CRSP case. However, for the ROP case, the predefined state constraint is violated.
    The simulation results displayed in Fig.\ref{simulation result of pendulum} have verified that the resulting approximate optimal control policy by solving Problem \ref{Optimal control problem} can efficiently deal with Problem \ref{Robust constraint control problem}.
 \subsection{Example 3: CRSP of a 2-DoF Robot Manipulator} \label{simulation 2-DoF}
 To further demonstrate the real-time performance of our proposed control framework,
 we consider a target (set-point) tracking control problem of a 2-DoF robot manipulator in this section. The objective is to drive the robot manipulator to reach the target point under additive disturbances, input and state constraints. For brevity, denoting $q = [q_1, q_2]^{\top}$, $\dot{q} = [\dot{q}_1, \dot{q}_2]^{\top}$, $c_{2} = \cos q_{2}$,  $s_{2} = \sin q_{2}$. Then, the Euler-Lagrange (E-L) equation of a 2-DoF robot manipulator follows \cite{li2021concurrent}
\begin{equation*} \label{2DoF model}
M(q) \Ddot{q}+C(q,\dot{q})\dot{q}=\tau,
\end{equation*}
where the inertia matrix is $\small{M(q) = \begin{bmatrix}
    p_{1}+2p_{3}c_2 & p_{2}+p_{3}c_2 \\
    p_{2}+p_{3}c_2 & p_{2} \\
\end{bmatrix}}$, and the matrix of centrifugal and Coriolis terms follows $\small{C(q,\dot{q}) = \begin{bmatrix}
    -p_{3}\dot{q}_{2}s_2 & -p_{3}(\dot{q}_{1}+\dot{q}_{2})s_2\\
    p_{3}\dot{q}_{1} s_2 & 0\\
\end{bmatrix}}$.
Let $[x_1,x_2,x_3,x_4]=[q_1,q_2,\dot{q}_1,\dot{q}_2]^{\top}$,
the above E-L equation can be written in the state-space form as \eqref{original sys}, where 
$f(x) = [x_3,x_4,(M^{-1}(-C)[x_3, x_4]^{\top})^{\top}]^{\top}$, and $g(x) = [[0,0]^{\top},[0,0]^{\top},(M^{-1})^{\top}]^{\top}$. Besides, with $k(x)=[[1,0]^{\top},[0,1]^{\top},0_{2 \times 2}]^{\top}$, we assume that the robot manipulator suffers an additive disturbance that follows $d(x)= [\delta_1 x_{1} \sin(x_{2}),\delta_2 x_{2} \cos(x_{1})]^{\top}$ , where $\delta_1, \delta_2 \in [-1,1]$.
Given $\left\|d(x)\right\| \leq \left\| x \right\|$, and $g^{\dagger}(x)k(x)d(x) = 0$, Assumptions \ref{bound of d and inverse term}-\ref{bound of decomposed disturbance part} can be satisfied by setting $d_{M}(x) = \left\| x \right\|$, and $\left\|l_{M}(x)\right\| = 0$.
The input constraints are considered as $\mathbb{U}_j = \left\{u_j \in \mathbb{R} : |u_j| \leq 3 \right\}$, $j =1,2$. The state constraint in a rectangular form has been considered in Section \ref{simulation pendulum}. Thus, to show the effectiveness of the proposed RS-SP term to deal with general state constraints, a circular state constraint $\mathbb{X}_3 = \left\{x \in \mathbb{R}^{2} : h_3(x_1,x_2) = x^2_1+x^2_2-1<0\right\}$ is considered here.
The solution procedure for the robot manipulator's CRSP is similar to Section \ref{simulation pendulum}. Thus, the detailed explanation is omitted here for page limits. 
To fully demonstrate the effectiveness of our proposed control framework to address state constraints even under input saturation and environmental disturbances, the phase plot under multiple initial positions is shown in Fig.\ref{circe constriant phase plot}, where the robot manipulator is driven to reach the target point (zero point) while obeying the predefined state constraints. We observe that when states approach to the constraint boundary, they will be driven back to safe states under the proposed control framework.
\begin{figure}[!t]  
\centerline{\includegraphics[width=3.8in]{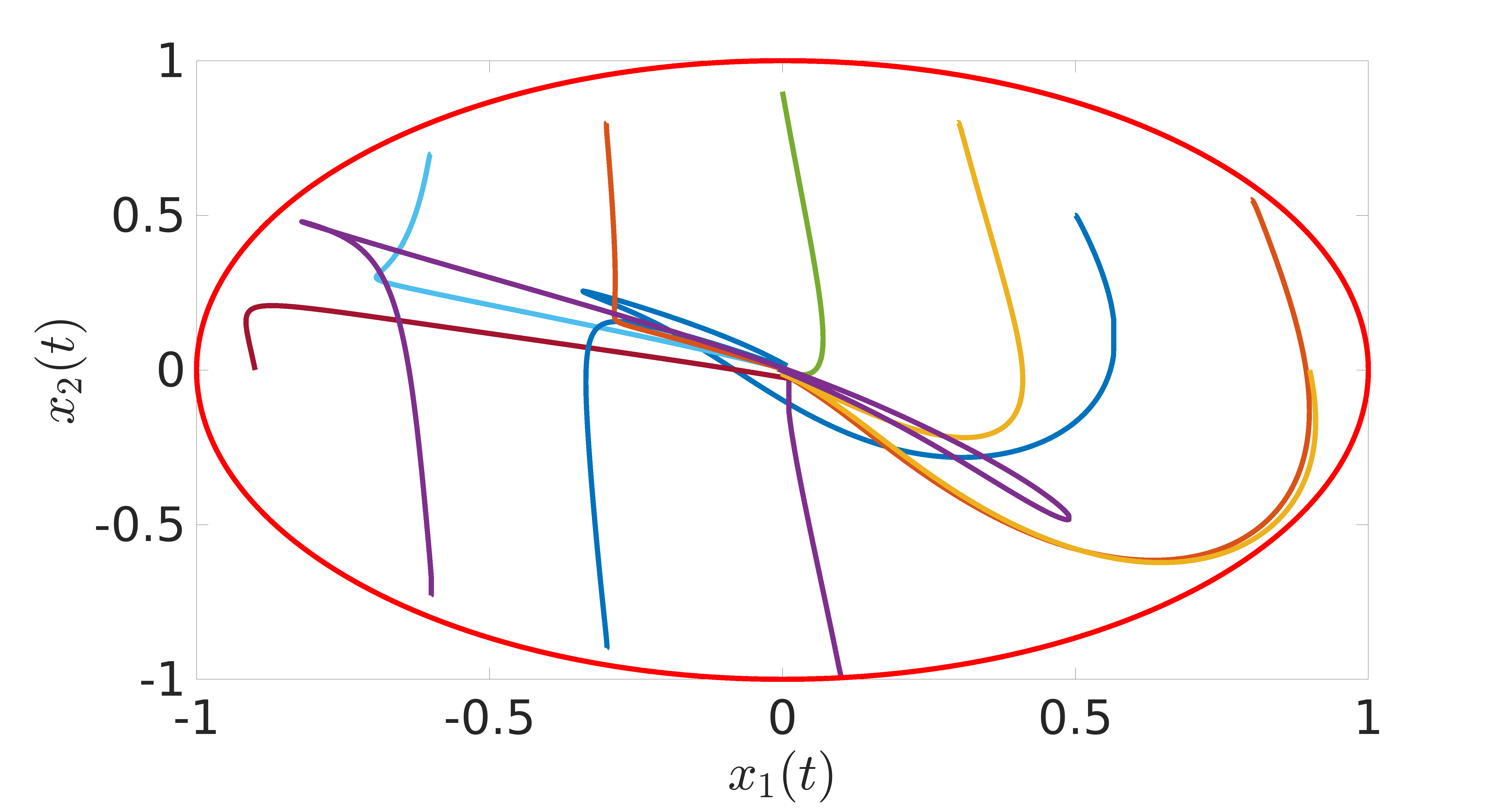}}
\caption{The phase plot of states $x_1$ and $x_2$ under different initial values for the 2-DoF robot manipulator. The red circle stands for the state constraint boundary.}
\label{circe constriant phase plot}
\end{figure}

In this section, we choose the activation function as $\Phi(x) = [x^{2}_{1},x^{2}_{2},x^{2}_{3},x^{2}_{4},x_{1}x_{2},x_{1}x_{3},x_{1}x_{4},x_{2}x_{3},x_{2}x_{4},x_{3}x_{4}]^{\top}$.
The weight convergence result is displayed in Fig.\ref{fig of critic weight 2DoF}, where the weight convergence result under our proposed weight update law \eqref{w update law} achieves at $t= 12$ $\mathrm{s}$.
The aforementioned simulation results based on a 2-DoF robot manipulator validate that the off-policy weight update law \eqref{w update law} and the approximate optimal control policy \eqref{approximation u} fulfill real-time requirements.
    \begin{figure}[!t]  
    \centering
    \includegraphics[width=3.8in]{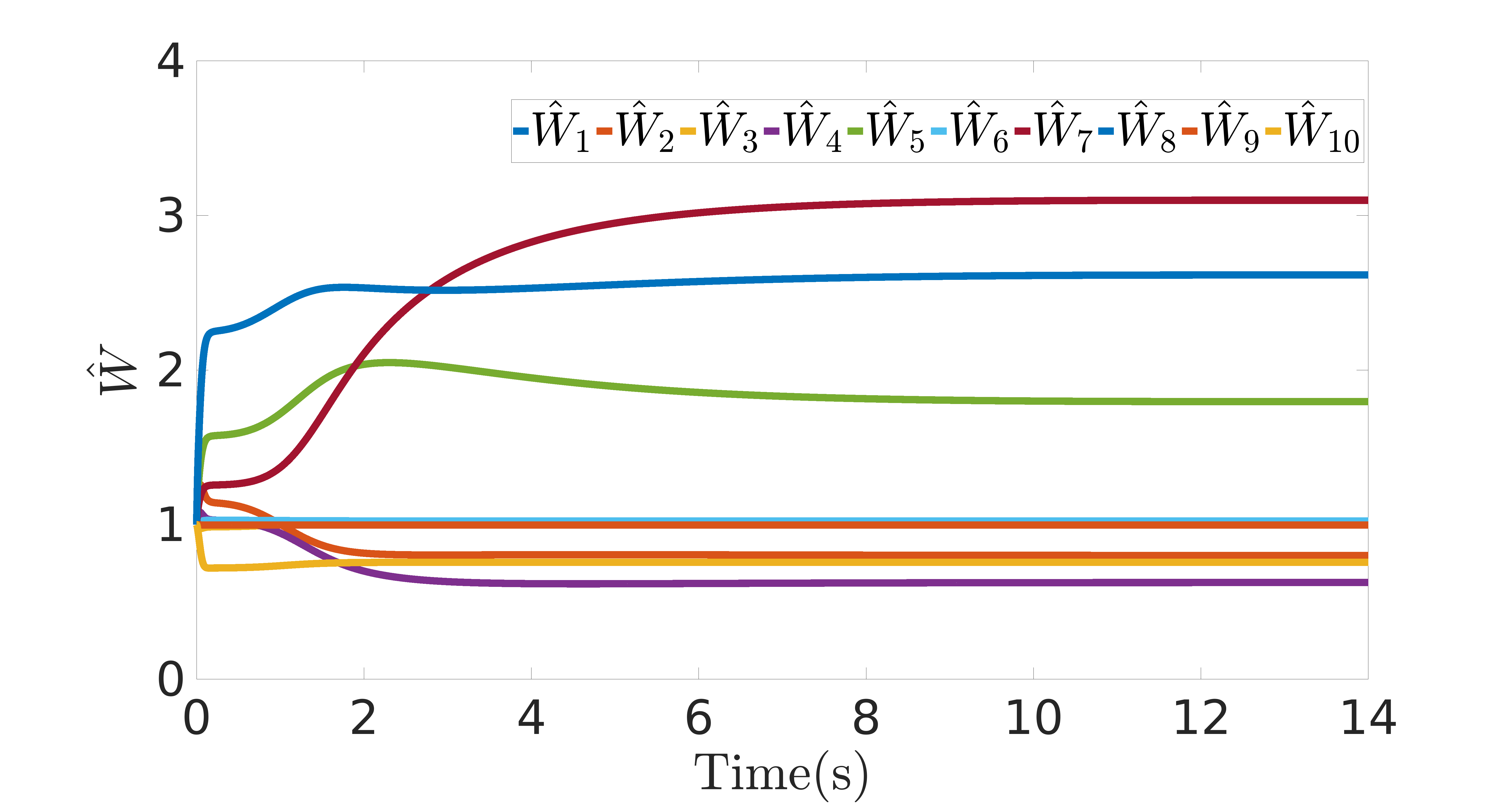}
    \caption{The evolution trajectory of the estimated weight $\hat{W}$ of the 2-DoF robot manipulator.}
    \label{fig of critic weight 2DoF}
    \end{figure}
    \section{Conclusion} \label{Section 5}
    An off-policy risk-sensitive RL-based control framework is proposed to stabilize a nonlinear system that subjects to additive disturbances, input and state constraints. 
    Firstly, the introduced pseudo control and the resulting auxiliary system enable us to address additive disturbances under an optimization framework.
    Then, risk-sensitive input and state penalty terms are incorporated into the cost function as optimization criteria, which allows us to tackle both input and state constraints in a long time-horizon. This helps to avoid the abrupt changes of control inputs that are unfavourable for the online learning process.
    The transformed OCP is solved by a single critic structure based ADP with an off-policy weight update law. The adopted single critic structure leads to computation simplicity and eliminates approximation errors caused by an actor NN. Besides, the exploitation of experience data to guarantee the weight convergence enables the proposed control strategy be applicable to practical applications.
    Multiple comparison results shown in the simulation part illustrate the effectiveness of the proposed control framework.  
     
\appendices
\section{Proof of Theorem 1} \label{Theorem 1}
\begin{proof}\label{proof for equivalence}
    \emph{(i) Proof of stability}. As for $V^{*}(x)$ defined as \eqref{optimal cost function}, we know that when $x =0$, $V^{*}(x) =0$, and $V^{*}(x) >0$ for $\forall x \ne 0$. 
    Thus, it can serve as a Lyapunov function candidate for stability proofs. Taking time derivative of $V^{*}(x)$ along the system \eqref{original sys} yields
    \begin{equation} \label{equalivance 1}
    \begin{aligned}
     \dot{V}^{*} &= \nabla {V^*}^{\top}(f(x)+g(x)\mu^{*}(x)+k(x)d(x)) \\
             &= \nabla {V^*}^{\top}(f(x)+g(x)\mu^{*}(x)+h(x)\nu^{*}(x))\\
             &+\nabla {V^*}^{\top} g(x) g^{\dagger}(x)k(x)d(x) + \nabla {V^*}^{\top} h(x)(d(x)-\nu^{*}(x)).
    \end{aligned}
    \end{equation}
    In light of \eqref{HJB equation}, we can get
    \begin{equation} \label{equalivance 2}
    \begin{aligned}
    &\nabla {V^*}^{\top}(f(x)+g(x)\mu^{*}(x)+h(x)\nu^{*}(x)) =- \mathcal{W}(\mu^{*}(x))\\
    & -\mathcal{L}(x)-\rho {\nu^*}^{\top}(x)\nu^{*}(x)-l^{2}_{M}(x) - \rho d^{2}_{M}(x).
   \end{aligned}
    \end{equation}
    From \eqref{optimal u}, we can get
    \begin{equation} \label{equalivance 3}
    \nabla {V^*}^{\top}g(x) = -2\beta R \tanh^{-1}(\mu^{*}(x) / \beta).
    \end{equation}
    Based on \eqref{optimal v}, the following equation establishes
    \begin{equation} \label{equalivance 4}
    \nabla {V^*}^{\top}h(x) = -2 \rho \nu^{*}(x).
    \end{equation}
    Substituting \eqref{equalivance 2}, \eqref{equalivance 3} and \eqref{equalivance 4} into \eqref{equalivance 1} yields
    \begin{equation} \label{equalivance 5}
        \begin{aligned}
             \dot{V}^{*} &=-\mathcal{W}(\mu^{*}(x))-\mathcal{L}(x)-\rho {\nu^*}^{\top}(x)\nu^{*}(x)-l^{2}_{M}(x) \\
             &- \rho d^{2}_{M}(x) -2\beta R \tanh^{-1}(\mu^{*}(x) / \beta)g^{\dagger}(x)k(x)d(x)\\
                       & - 2 \rho {\nu^*}^{\top}(x) d(x)
                        + 2 \rho {\nu^*}^{\top}(x)\nu^{*}(x).
        \end{aligned}
    \end{equation}
    By setting $\varsigma_j = \tanh^{-1}(\tau_{j} / \beta)$, we can get
    \begin{equation} \label{equalivance 6 of W}
    \begin{aligned}
    \mathcal{W}(\mu^{*}(x)) &= 2 \beta \sum_{j=1}^{m}\int_{0}^{\mu^{*}_{j}} R_j \tanh^{-1}(\tau_{j} / \beta)\,d\tau_{j} \\
    &=2 \beta^2 \sum_{j=1}^{m}\int_{0}^{\tanh^{-1}(\mu^{*}_{j} / \beta)}R_j \varsigma_j (1-\tanh^2 (\varsigma_j))\,d \varsigma_j \\
                     &= \beta^2 \sum_{j=1}^{m} R_j (\tanh^{-1}(\mu^{*}_{j}(x) / \beta ))^2 - \epsilon_{t},
    \end{aligned}
    \end{equation}
where $\epsilon_{t} := 2 \beta^2 \sum_{j=1}^{m} \int_{0}^{\tanh^{-1}(\mu^{*}_{j}(x) / \beta )} R_j \varsigma_j \tanh^2(\varsigma_j)\,d\varsigma_j$. 
Based on the integral mean-value theorem, there exist a series of $\theta_j \in [0,\tanh^{-1}(\mu^{*}_{j}(x) / \beta], j = 1, \cdots, m$, such that
    \begin{equation} \label{epsilon_u_1}
        \epsilon_{t} =  2 \beta^2 \sum_{j=1}^{m}  R_j \tanh^{-1}(\mu^{*}_{j}(x) / \beta ) \theta_j \tanh^2(\theta_j).
    \end{equation}
Bearing in mind the relation \eqref{equalivance 3} and the fact  $0<\tanh^2(\theta_j)\leq 1$, it follows that
    \begin{equation} \label{epsilon_u_2}
        \begin{aligned}
            \epsilon_{t} & \leq  2 \beta^2 \sum_{j=1}^{m} R_j \tanh^{-1}(\mu^{*}_{j}(x) / \beta )\theta_j \\
            & \leq 2 \beta^2 \sum_{j=1}^{m }R_j (\tanh^{-1}(\mu^{*}_{j}(x) / \beta))^2\\
            & = \frac{1}{2} \nabla {V^*}^{\top}g(x) R ^{-1} g^{\top}(x)\nabla V^{*}.
        \end{aligned}
    \end{equation}
    According to the definition of the admissible policy \cite{abu2005nearly}, $V^*$ is finite. Moreover, there exists $w_{M} > 0$ such that $\left\|\nabla V^{*}\right\| \leq \omega_{M}$. Based on Assumption \ref{bound of fg},  we can rewrite \eqref{epsilon_u_2} as
    \begin{equation} \label{epsilon_u_3}
        \begin{aligned}
            \epsilon_{t}  \leq b_{\epsilon_{t}}.
        \end{aligned}
    \end{equation}
    where $b_{\epsilon_{t}} := \frac{1}{2} \left\| R^{-1} \right\| g^2_M \omega^2_{M}$.
    Based on Assumption \ref{bound of d and inverse term}, the following equations establish
    \begin{equation} \label{Young inequality 1} 
        \begin{aligned}
            &-2\beta R \tanh^{-1}(\mu^{*}(x) / \beta)g^{\dagger}(x)k(x)d(x) \\
           & \leq \left\|\beta R \tanh^{-1}(\mu^{*}(x) / \beta) \right\|^2
             + \left\| g^{\dagger}(x)k(x)d(x) \right\|^2\\
             &\leq \beta^2 \sum_{j=1}^{m}R^{2}_j(\tanh^{-1}(\mu^{*}(x) / \beta)) ^2 + l_{M}^2(x),
        \end{aligned}
    \end{equation}
and
    \begin{equation} \label{Young inequality 2}
        \begin{aligned}
            &- 2 \rho {\nu^*}^{\top}(x) d(x) 
             \leq \rho \left\| \nu^{*}(x)  \right\|^2
            + \rho \left\| d(x)  \right\|^2 \\
            &\leq \rho \left\| \nu^{*}(x)  \right\|^2 + \rho d_{M}^2(x).
        \end{aligned}
    \end{equation}
    By substituting \eqref{equalivance 6 of W}, \eqref{epsilon_u_3}, \eqref{Young inequality 1} and \eqref{Young inequality 2} into \eqref{equalivance 5}, we have
    \begin{equation} \label{equalivance 7}
        \begin{aligned}
            \dot{V}^{*} &\leq -\mathcal{L}(x) + 2 \rho {\nu^*}^{\top}(x)\nu^{*}(x)+b_{\epsilon_{t}}\\
            &+\beta^2 \sum_{j=1}^{m}(R^{2}_j-R_j)(\tanh^{-1}(\mu^{*}(x) / \beta)) ^2 \\
           & =-\mathcal{L}(x) + 2 \rho {\nu^*}^{\top}(x)\nu^{*}(x)+ \epsilon_{s}.
        \end{aligned}
    \end{equation}
    where $\epsilon_{s} := b_{\epsilon_{t}}+\beta^2 \sum_{j=1}^{m}(R^{2}_j-R_j)(\tanh^{-1}(\mu^{*}(x) / \beta)) ^2$.
    Thus, $\dot{V}^{*}<0$ establishes, if the condition~\eqref{equalivance condition} holds.
    It yields that the optimal control policy $\mu^{*}(x)$ robustly stabilizes the system \eqref{original sys}.
    
    \emph{(ii) Proof of input and state constraint satisfaction}. 
    Denote $V^{*}(0)$ the value of the Lyapunov function candidate $V^{*}$ at $t = 0$. According to the definition of admissible control policies, $V^{*}(0)$ is a bounded function. If $\mathcal{L}(x) > 2\rho {\nu^*}^{\top}(x)\nu^{*}(x)+\epsilon_{s}$, $\dot{V}^{*} <0$ establishes, which means that $V^{*}(t) < V^{*}(0)$, $\forall t$. The boundness of $V^{*}(t)$ implies that state constraints will not be violated; otherwise, according to Definition \ref{risk-sensitive SPF}, $V^{*}(t) \to \infty$ if any state constraint violations happens.
     Since the hyperbolic tangent function satisfies $-1 \leq \tanh(\cdot) \leq 1$, the optimal control policy in \eqref{optimal u} follows  $-\beta \leq  \mu^{*}(x) \leq \beta$, i.e., inputs are confined into the safety set \eqref{input saturation}. The proof provided here means that the optimal control policy $\mu^{*}(x)$ for the system \eqref{original sys} guarantees satisfaction of both constraints in terms of system states and control inputs.
    \end{proof}

\section{Proof of Theorem 3} \label{Theorem 3}
\begin{proof}\label{proof for optimal problem results}
    Consider the following candidate Lyapunov function
    \begin{equation} \label{Lya function stability}
    J =  V^{*}(x) + \frac{1}{2} \Tilde{W}^{\top} \Gamma^{-1} \Tilde{W}.
    \end{equation}
    Taking time derivative of \eqref{Lya function stability} along the system \eqref{auxiliary system} yields
    \begin{equation} \label{stability dV}
       \dot{J}  = \dot{L}_{V}+\dot{L}_{W}.
    \end{equation}
    where $\dot{L}_{V} = \dot{V}^{*}(x)$ and $\dot{L}_{W}= \Tilde{W}^{\top} \Gamma^{-1} \dot{\hat{W}}$.
    
    \emph{The first term $\dot{L}_{V}$ follows}
    \begin{equation} \label{stability dLv}
        \begin{aligned}
            \dot{L}_{V} &= \nabla {{V}^{*}}^{\top}(f(x)+g(x)\hat{u}(x)+h(x)\hat{v}(x)) \\ 
                        &= \nabla {{V}^{*}}^{\top}(f(x)+g(x)\mu^{*}(x)+h(x)\nu^{*}(x))\\
                        &+\nabla {{V}^{*}}^{\top}g(x)(\hat{u}(x)-\mu^{*}(x))+\nabla{{V}^{*}}^{\top}h(x)(\hat{v}(x)-\nu^{*}(x)). \\ 
        \end{aligned}
    \end{equation}
    According to \eqref{equalivance 2}, \eqref{equalivance 3} and \eqref{equalivance 4}, \eqref{stability dLv} can be rewritten as
    \begin{equation} \label{stability dLv 1}
        \begin{aligned}
            \dot{L}_{V} &= -\mathcal{L}(x)- \mathcal{W}(\mu^{*}(x))-\rho \nu^{*}(x)^{\top}\nu^{*}(x)-l^{2}_{M}(x) \\            
            &- \rho d^{2}_{M}(x) -2\beta R \tanh^{-1}(\mu^{*}(x) / \beta)(\hat{u}(x)-\mu^{*}(x))\\
            &-2\rho {\nu^*}^{\top}(x)(\hat{v}(x)-\nu^{*}(x)). \\ 
        \end{aligned}
    \end{equation}
    Besides, we can get 
    \begin{equation} \label{stability dLv 2} 
        \begin{aligned}
    &-2\beta R \tanh^{-1}(\mu^{*}(x) / \beta)(\hat{u}(x)-\mu^{*}(x)) \\
    & \leq \beta^2 \left\| R \tanh^{-1}(\mu^{*}(x) / \beta) \right\|^2 
         + \left\| \hat{u}(x)-\mu^{*}(x) \right\|^2\\
       & \leq \beta^2 \sum_{j=1}^{m}R^{2}_j(\tanh^{-1}(\mu^{*}_{j}(x) / \beta ))^2 + \left\| \hat{u}(x)-\mu^{*}(x) \right\|^2.
        \end{aligned}
    \end{equation}
    Based on \eqref{equalivance 6 of W}-\eqref{equalivance 7}, the following equation also establishes
    \begin{equation} \label{stability dLv 3} 
        \begin{aligned}
        &- \mathcal{W}(\mu^{*}(x))-2\beta\tanh^{-1}(\mu^{*}(x) / \beta)(\hat{u}(x)-\mu^{*}(x)) \\
        &\leq \beta^2 \sum_{j=1}^{m} (R^2_j-R_j)(\tanh^{-1}(\mu^{*}_{j}(x)/\beta))^2+ b_{\epsilon_{t}} \\
        & + \left\| \hat{u}(x)-\mu^{*}(x) \right\|^2
         \leq \epsilon_{s}+\left\| \hat{u}(x)-\mu^{*}(x) \right\|^2 .
        \end{aligned}
    \end{equation}
    Substituting \eqref{stability dLv 3} into \eqref{stability dLv 1} yields
    \begin{equation} \label{stability dLv 4}
        \begin{aligned}
            \dot{L}_{V} &\leq -\mathcal{L}(x)-\rho \nu^{*}(x)^{\top}\nu^{*}(x)-l^{2}_{M}(x) - \rho d^{2}_{M}(x)+\epsilon_{s} \\
            & +\left\| \hat{u}(x)-\mu^{*}(x) \right\|^2 -2\rho {\nu^*}^{\top}(x)(\hat{v}(x)-\nu^{*}(x)) \\ 
            &= -\mathcal{L}(x)-l^{2}_{M}(x) - \rho d^{2}_{M}(x)+\epsilon_{s}-\rho\hat{v}^{\top}(x)\hat{v}(x)\\
            &+\left\| \hat{u}(x)-\mu^{*}(x) \right\|^2+\rho \left\| \hat{v}(x)-\nu^{*}(x) \right\|^2.
        \end{aligned}
    \end{equation}
    \emph{As for $\rho\hat{v}^{\top}(x)\hat{v}(x)$ in \eqref{stability dLv 4}}, according to \eqref{approximation v},
    \begin{equation} \label{pvv}
        \begin{aligned}
            &\rho\hat{v}^{\top}(x)\hat{v}(x) = \frac{1}{4 \rho}\hat{W}^{\top} \nabla \Phi(x) h(x)h^{\top}(x) \nabla \Phi^{\top}(x)\hat{W}\\
            &= \frac{1}{4 \rho}(W^{*}+\Tilde{W})^{\top} \nabla \Phi(x) h(x)h^{\top}(x) \nabla \Phi^{\top}(x)(W^{*}+\Tilde{W})\\
            &= \frac{1}{4 \rho}{W^*}^{\top} \mathscr{H} W^{*}
            +\frac{1}{4 \rho}\Tilde{W}^{\top} \mathscr{H} \Tilde{W}
            +\frac{1}{2 \rho} {W^*}^{\top} \mathscr{H}\Tilde{W}.
        \end{aligned}
    \end{equation}
    where $\mathscr{H} = \nabla \Phi(x) h(x)h^{\top}(x) \nabla \Phi^{\top}(x)$.
    
    \emph{As for $\rho \left\| \hat{v}(x)-\nu^{*}(x) \right\|^2$in \eqref{stability dLv 4}}, according to \eqref{approximation v},
    \begin{equation} \label{pvv2}
        \begin{aligned}
           & \rho \left\| \hat{v}(x)-\nu^{*}(x) \right\|^2
            = \rho \left\| \frac{1}{2 \rho} h^{\top}(x)\nabla \Phi(x)\Tilde{W} \right\|^2 \\
            & = \frac{1}{4 \rho}\Tilde{W} ^{\top} \mathscr{H}\Tilde{W}. 
        \end{aligned}
    \end{equation}
    For simplicity, denote $\mathscr{G}^* = \frac{1}{2\beta} R^{-1} g^{\top}(x)\nabla \Phi^{\top}(x)W^{*}$ and $\hat{\mathscr{G}} = \frac{1}{2\beta} R^{-1} g^{\top}(x)\nabla \Phi^{\top}(x)\hat{W}$, $\hat{\mathscr{G}} = [\hat{\mathscr{G}}_1,\cdots,\hat{\mathscr{G}}_m] \in \mathbb{R}^{m}$ with $\hat{\mathscr{G}}_j \in \mathbb{R}, j=1,\cdots,m$. 
    Based on \eqref{optimal u} and \eqref{approximation u}, the Taylor series of $\tanh(\mathscr{G}^*)$ follows
    \begin{equation} \label{Taylor series}
        \begin{aligned}
            \tanh(\mathscr{G}^*) &= \tanh(\hat{\mathscr{G}}) +\frac{\partial \tanh(\hat{\mathscr{G}})}{\partial\hat{\mathscr{G}}}(\mathscr{G}^*-\hat{\mathscr{G}})+O((\mathscr{G}^*-\hat{\mathscr{G}})^2)\\
            &=\tanh(\hat{\mathscr{G}})-\frac{1}{2 \beta}(I_{m \times m} - \mathscr{D}(\hat{\mathscr{G}})) R^{-1} g^{\top}(x)\\
            &~~\nabla \Phi^{\top}(x)\Tilde{W} + O((\mathscr{G}^*-\hat{\mathscr{G}})^2),
        \end{aligned}
    \end{equation}
    where $\mathscr{D}(\hat{\mathscr{G}}) = \diag{\tanh^2(\hat{\mathscr{G}}_1),\cdots,\tanh^2(\hat{\mathscr{G}}_m)}$, $O((\mathscr{G}^*-\hat{\mathscr{G}})^2)$ is a higher order term of the Taylor series. 
    By following  \cite[Lemma 1]{yang2016online}, the higher order term is bounded as
    \begin{equation}\label{bound of higher order term}
        \begin{aligned}
            \left\| O((\mathscr{G}^*-\hat{\mathscr{G}})^2)  \right\| & \leq 2\sqrt{m}+\frac{1}{\beta} \left\| R^{-1} \right\| g_M {b_\Phi}_x \left\|\Tilde{W} \right\|.
        \end{aligned}
    \end{equation}
    Using \eqref{optimal u}, \eqref{approximation u} and \eqref{Taylor series}, we get 
    \begin{equation}\label{u-u*-a}
        \begin{aligned}
            \hat{u}(x)-\mu^{*}(x) &= \beta(\tanh(\mathscr{G}^*)-\tanh(\hat{\mathscr{G}}))+\epsilon^*_u \\
            &=-\frac{1}{2}(I_{m \times m} - \mathscr{D}(\hat{\mathscr{G}})) R^{-1} g^{\top}(x)\nabla \Phi^{\top}(x)\Tilde{W} \\
            &+ \beta O((\mathscr{G}^*-\hat{\mathscr{G}})^2)+\epsilon^*_u.\\
        \end{aligned}
    \end{equation}
      where $\epsilon^*_u := \beta \tanh(\frac{1}{2\beta} R^{-1} g^{\top}(x)(\nabla \Phi^{\top}(x)W^*+\nabla \epsilon))-\beta \tanh(\frac{1}{2\beta} R^{-1} g^{\top}(x)\nabla \Phi^{\top}(x)W^*)$, and assuming that it is bounded by $\left\|\epsilon^*_u \right\| \leq b_{\epsilon^*_u}$.
    
    \emph{As for $\left\| \hat{u}(x)-\mu^{*}(x) \right\|^2$ in \eqref{stability dLv 4}}, since $\left\|I_{m \times m} - \mathscr{D}(\hat{\mathscr{G}})\right\| \leq 2$ \cite{yang2016online}, by combining \eqref{bound of higher order term} and \eqref{u-u*-a}, we get
    \begin{equation}\label{u-u* abs 2}
        \begin{aligned}
            &\left\| \hat{u}(x)-\mu^{*}(x) \right\|^2 \leq 3 \beta^2 \left\| O((\mathscr{G}^*-\hat{\mathscr{G}})^2)\right\|^2 +3\left\| \epsilon^*_u \right\|^2\\
            & +3\left\| -\frac{1}{2}(I_{m \times m} - \mathscr{D}(\hat{\mathscr{G}})) R^{-1} g^{\top}(x)\nabla \Phi^{\top}(x)\Tilde{W}\right\|^2 \\
            & \leq 6 \left\|R^{-1}\right\| ^2 g_M^2 {b^2_\Phi}_x \left\|\Tilde{W}\right\|^2+ 12m\beta^2 + 3b^2_{\epsilon^*_u} \\
            & + 12 \beta \sqrt{m}  \left\|R^{-1}\right\| g_M {b_\Phi}_x \left\|\Tilde{W}\right\| .
        \end{aligned}
    \end{equation}
    Substituting \eqref{pvv}, \eqref{pvv2}, \eqref{u-u* abs 2} into \eqref{stability dLv 4} yields
    \begin{equation} \label{stability dLv 5}
        \begin{aligned}
            \dot{L}_{V} &\leq -\frac{1}{2 \rho} {W^*}^{\top} \mathscr{H}\Tilde{W}-\mathcal{L}(x)-l^{2}_{M}(x)- \rho d^{2}_{M}(x)\\
            & -\frac{1}{4 \rho} {W^*}^{\top} \mathscr{H}{W^*}  +\epsilon_{s}
            +6 \left\|R^{-1}\right\|^2 g_M^2 {b^2_\Phi}_x \left\|\Tilde{W}\right\|^2 \\
            &+ 12m\beta^2 + 3b^2_{\epsilon^*_u} 
            + 12 \beta \sqrt{m} \left\|R^{-1}\right\| g_M {b_\Phi}_x \left\|\Tilde{W}\right\|.
        \end{aligned}
    \end{equation}
    \emph{As for the second term $\dot{L}_{W}$}, based on \eqref{w update law} and \eqref{dVcl},
    \begin{equation} \label{stability dLw}
        \begin{aligned}
            \dot{L}_{W} 
                &\leq - \Tilde{W}^{\top}X\Tilde{W}+\Tilde{W}^{\top}\epsilon_{er}.
        \end{aligned}
    \end{equation}
    
    \emph{Finally, as for $\dot{J}$}, substituting \eqref{stability dLv 5} and \eqref{stability dLw} into \eqref{stability dV}, based on the fact that $\left\|W^{*}\right\| \leq b_{W^{*}}$, $\left\|\nabla \Phi(x)\right\| \leq b_{\Phi_x}$, $\left\|h(x)\right\| \leq h_{M}$, we can get
    \begin{equation} \label{stability final result}
        \begin{aligned}
            \dot{J} & \leq -\mathcal{L}(x)-l^{2}_{M}(x) - \rho d^{2}_{M}(x)-\frac{1}{4 \rho} {W^*}^{\top} \mathscr{H}{W^*} \\
            & -\Tilde{W}^{\top}X\Tilde{W}+M\Tilde{W}+6 \left\|R^{-1}\right\|^2 g_M^2 {b^2_\Phi}_x \left\|\Tilde{W}\right\|^2 \\
            &+ 12 \beta \sqrt{m} \left\|R^{-1}\right\| g^2_M {b^2_\Phi}_x \left\|\Tilde{W}\right\| + 12m\beta^2 + 3b^2_{\epsilon^*_u}
            +\epsilon_{s}\\
            &\leq  -\mathcal{L}(x)-l^{2}_{M}(x) - \rho d^{2}_{M}(x)-\frac{1}{4 \rho} {W^*}^{\top} \mathscr{H}{W^*}\\
            &-(\lambda_{\min}(B)-6\left\| R^{-1}\right\|^2 g^2_M b^2_{{\Phi}_x }) \left\| \Tilde{W}\right\|^2 +12m\beta^2+ 3b^2_{\epsilon^*_u}\\
            &+(12 \beta \sqrt{m}\left\| R^{-1}\right\|g^2_M {b^2_\Phi}_x +b_M)\left\| \Tilde{W}\right\|
            +\epsilon_{s}\\
            &= -\mathcal{A}-\mathcal{B}\left\| \Tilde{W}\right\|^2+\mathcal{C}\left\| \Tilde{W}\right\|+\mathcal{D},
        \end{aligned} 
    \end{equation}
    where $M = \epsilon_{er}-  \frac{1}{2 \rho}{W^*}^{\top}\mathscr{H} $, and there exists $b_M := b_{\epsilon_{er}}+\frac{1}{2 \rho}  {b^2_\Phi}_x h^2_M b_{W^{*}} \in \mathbb{R}^+$ such that $\left\|M \right\| \leq  b_M$; 
    $\mathcal{A} = \mathcal{L}(x)+l^{2}_{M}(x) + \rho d^{2}_{M}(x)+\frac{1}{4 \rho} {W^*}^{\top} \mathscr{H}{W^*}$ is positive definite; 
    $\mathcal{B} = \lambda_{\min}(B)-6 \left\| R^{-1}\right\|^2 g^2_M b^2_{{\Phi}_x}$, 
    $\mathcal{C} = 12 \beta \sqrt{m}\left\| R^{-1}\right\|g^2_M {b^2_\Phi}_x +b_M$ and 
    $\mathcal{D} = 12m\beta^2+3b^2_{\epsilon^*_u}+\epsilon_{s}$.
    
    Let the parameters be chosen such that $\mathcal{B} > 0$. Since $\mathcal{A}$ is positive definite, the above Lyapunov derivative is negative if
    \begin{equation} \label{negative condition}
        \begin{aligned}
        \left\| \Tilde{W} \right\| > \frac{\mathcal{C}}{2\mathcal{B}}+\sqrt{\frac{\mathcal{C}^2}{4\mathcal{B}^2}+\frac{\mathcal{D} }{\mathcal{B}}}.
        \end{aligned}
    \end{equation}
    
    Thus, the critic weight learning error converges to the residual set defined as
    \begin{equation} \label{compact set}
        \begin{aligned}
        \Tilde{\Omega}_{\Tilde{W}} = \{\Tilde{W} | \left\| \Tilde{W} \right\| \leq \frac{\mathcal{C}}{2\mathcal{B}}+\sqrt{\frac{\mathcal{C}^2}{4\mathcal{B}^2}+\frac{\mathcal{D} }{\mathcal{B}}} \}.
        \end{aligned}
    \end{equation}

    \end{proof}

\bibliographystyle{IEEEtran}
\bibliography{bibtex/bib/IEEEexample}

\end{document}